\documentclass[entropy,article,submit,moreauthors,pdftex,10pt,a4paper]{mdpi} 
\firstpage{1} %%MDPI internal note: new layout%%
\articlenumber{x}
\doinum{10.3390/------}
\pubvolume{xx}
\pubyear{2016}
\copyrightyear{2016}
\externaleditor{{Academic Editors: Gregg Jaeger and Andrei Khrennikov}}
\history{{Received: 4 March 2016; Accepted: 30 May 2016; Published: date}}
\pdfoutput=1

\usepackage[utf8]{inputenc}
\usepackage[english]{babel}

\usepackage[active]{srcltx}
\usepackage{soul,bm,microtype}
\usepackage[usenames,dvipsnames]{xcolor}

\usepackage{amsmath,amssymb,amsfonts}
\usepackage{mathrsfs,mathdots}
\usepackage[centercolon=true]{mathtools}
\usepackage{braket}
\usepackage{suffix,xspace,xifthen}
\usepackage{subcaption}
\usepackage{booktabs,bookmark}

\usepackage[capitalise]{cleveref}
\usepackage{soul}

\newcommand*{\swap}[2]{\let\temp#1 \let#1#2 \let#2\temp \let\temp\relax}

\swap{\epsilon}{\varepsilon}
\swap{\theta}{\vartheta}
\swap{\phi}{\varphi}
\renewcommand{\Re}{\operatorname{Re}}
\renewcommand{\Im}{\operatorname{Im}}
\DeclareMathOperator{\FF}{DFT}
\DeclareMathOperator{\Arg}{Arg}

\DeclarePairedDelimiter{\card}{\lvert}{\rvert}
\DeclarePairedDelimiter{\vol}{\lvert}{\rvert}
\DeclarePairedDelimiter{\ceil}{\lceil}{\rceil}
\DeclarePairedDelimiter{\floor}{\lfloor}{\rfloor}
\DeclarePairedDelimiter{\paren}{\lparen}{\rparen}
\DeclarePairedDelimiter{\sparen}{[}{]}

\newcommand*{\numberset}{\mathbb}
\newcommand*{\N}{\numberset{N}}
\newcommand*{\Z}{\numberset{Z}}

\newcommand*{\R}{\numberset{R}}
\newcommand*{\C}{\numberset{C}}
\newcommand{\Powset}[1]{\mathcal{P}\paren{#1}}
\WithSuffix\newcommand\Powset*[1]{\mathcal{P}\paren*{#1}}

\newcommand*{\im}{i}
\newcommand*{\ims}{i\mkern1mu}

\newcommand*{\KetBra}[2]{\Ket{#1}\!\Bra{#2}}
\newcommand*{\ketbra}[2]{\ket{#1}\!\bra{#2}}

\newcommand*{\dif}[1]{\,\mathrm{d}\mkern0mu#1}

\newcommand*{\od}[3][]{
	\ifinner
	\tfrac{\r{d}{^{#1}}#2}{\r{d}{#3^{#1}}}
	\else
	\dfrac{\r{d}{^{#1}}#2}{\r{d}{#3^{#1}}}
	\fi
}

\newcommand*{\pd}[3][]{
	\ifinner
	\tfrac{\partial{^{#1}}#2}{\partial{#3^{#1}}}
	\else
	\dfrac{\partial{^{#1}}#2}{\partial{#3^{#1}}}
	\fi
}

\renewcommand*{\vec}[1]{\boldsymbol{\mathrm{#1}}}

\newcommand*{\Hilbert}[1][H]{\mathcal{#1}}
\def\H{\Hilbert}
\newcommand*{\Hom}[1]{\mathrm{Hom}\paren{#1}}
\WithSuffix\newcommand\Hom*[1]{\mathrm{Hom}\paren*{#1}}
\newcommand*{\End}[1]{\mathrm{End}\paren{#1}}
\WithSuffix\newcommand\End*[1]{\mathrm{End}\paren*{#1}}
\newcommand*{\Aut}[1]{\mathrm{Aut}\paren{#1}}
\WithSuffix\newcommand\Aut*[1]{\mathrm{Aut}\paren*{#1}}
\newcommand*{\Bnd}[1]{\mathcal{B}\paren{#1}}
\WithSuffix\newcommand\Bnd*[1]{\mathcal{B}\paren*{#1}}
\newcommand*{\supp}[1]{\mathrm{supp}\paren{#1}}
\WithSuffix\newcommand\supp*[1]{\mathrm{supp}\paren*{#1}}
\newcommand*{\Br}[1][B]{\mathsf{#1}}

\newcommand*{\mvec}[1]{\begin{pmatrix}#1\end{pmatrix}}

\newcommand{\Cay}[2]{\Gamma(#1,#2)}

\def\qed{$\,\blacksquare$\par}
\def\<{\langle}
\def\>{\rangle}
\def\mat#1{{\boldsymbol{#1}}}
\def\v#1{\boldsymbol{\mathrm #1}} 
\def\Neigh{\mathcal N} 

\def\L2{{\mathcal L}_2}

\def\df#1#2 {\!\!\frac{\mathop{\mathrm{d}#1}}{#2}\,}
\newcommand*{\HI}[2][\bk]{H^{#2}(#1)}
\def\HIW{\HI{W}}

\def\bvec#1{\mathbf{#1}}
\def\bk{\bvec k}
\def\bh{\bvec h}

\def\bn{\bvec n}

\def\bx{\bvec x}

\def\bn{\bvec n}
\def\br{\bvec r}

\def\bv{\bvec v}

\def\bsi{\boldsymbol\sigma}

\def\bga{\boldsymbol\gamma}

\def\TR{\intercal}
\def\EM{\mathcal{E}}

\def\r#1{\mathrm{#1}}

\newcommand{\intentionalspace}{\par}
\let\intentionalspace\relax
\newcommand*{\DefaultAcronymStyle}[1]{{#1}}
\newcommand*{\Acronym}[2][\DefaultAcronymStyle]{\expandafter\def\csname #2\endcsname{#1{#2}\xspace}}
\WithSuffix\newcommand\Acronym*[3][\DefaultAcronymStyle]{\expandafter\def\csname #2\endcsname{#1{#3}\xspace}}

\Acronym{QCA}
\Acronym{QW}
\Acronym{RW}
\Acronym{BCC}
\Acronym{PDE}

\Acronym{DFT}
\Acronym{FFT}
\Acronym{QFT}

\Acronym*[\emph]{ie}{i.e.}
\Acronym*[\emph]{eg}{e.g.}

 \theoremstyle{mdpi}
 \newcounter{thm}
 \newcounter{ex}
 \newcounter{re}
 \newcounter{def}
 \newtheorem{Proposition}[thm]{Proposition}
 \newtheorem{Definition}[def]{Definition}

 \theoremstyle{mdpidefinition}

 \newtheorem{Remark}[re]{Remark}
 
 \renewcommand{\qedsymbol}{$\blacksquare$}

 \Title{Discrete Time Dirac Quantum Walk in 3+1~Dimensions}

\Author{Giacomo Mauro D'Ariano $^{1,}$*, Nicola Mosco $^{2}$, Paolo Perinotti $^{2}$ and Alessandro Tosini $^{2}$}
\address{%
$^{1}$ \quad QUIT group, Dipartimento di Fisica, via Bassi 6, Pavia 27100, Italy \\
$^{2}$ \quad INFN Gruppo IV, Sezione di Pavia, via Bassi 6, Pavia 27100, Italy; nicola.mosco01@ateneopv.it (N.M.); {paolo.perinotti@unipv.it (P.P.); alessandro.tosini@unipv.it (A.T.)}\\
 }

\corres{Correspondence: dariano@unipv.it; {Tel.: +39-0382-987-484}}

\abstract{In this paper we consider quantum walks whose evolution
  converges to the Dirac equation in the limit of small
  wave-vectors. We show exact Fast Fourier implementation of the Dirac
  quantum walks in one, two, and three space dimensions. The behaviour
  of particle states---defined as states smoothly peaked in some
  wave-vector eigenstate of the walk---is described by an approximated
  dispersive differential equation that for small wave-vectors gives
  the usual Dirac particle and antiparticle kinematics. The accuracy
  of the approximation is provided in terms of a lower bound on the
  fidelity between the exactly evolved state and the approximated
  one. The jittering of the position operator expectation value for
  states having both a particle and an antiparticle component is
  analytically derived and observed in the numerical implementations.}
\keyword{quantum walk; Dirac equation}

\PACS{03.67.Ac, 03.67.-a}
\begin{document}

%%%%%%%%%%%%%%%%%%%%%%%%%%%%%%%%%%%%%%%%%%
%% Sections that are not mandatory are listed as such. The section titles given are for Articles. Review papers and other article types have a more flexible structure. 

%% Only for the journal Gels: Please place the Experimental Section after the Conclusions

%%%%%%%%%%%%%%%%%%%%%%%%%%%%%%%%%%%%%%%%%%
%\setcounter{section}{-1} %% Remove this when starting to work on the template.
%\section{How to Use this Template}
%
%The template details the sections that can be used in a manuscript. Sections that are not mandatory are listed as such. The section titles given are for Articles. Review papers and other article types have a more flexible structure. 
%Remove this paragraph. For any questions, please contact the editorial office of the journal or support@mdpi.com. For LaTeX related questions please contact Janine Daum at latex-support@mdpi.com.

\section{Introduction}

Thinking about the discrete evolution of physical systems, the most
natural example is certainly a particle moving on a lattice. A
(classical) \emph{random walk} is exactly the description of a
particle which moves in discrete time steps and with certain
probabilities from one lattice position to the neighboring lattice
positions.  These models have gained increasing attention, showing
several applications in the fields of mathematics, physics, chemistry,
computer science, natural sciences, and economics
\cite{VanKampen:1992a,Weiss:2005a,Cox:1962a}.  A quantum version of
such a random walk---denoted \emph{quantum walk} (\QW)---was first
introduced in \cite{Aharonov:1993aa}, where the motion (right or left)
of a spin-$1/2$ particle is decided by a measurement of the
$z$-component of its spin. Subsequently, the measurement was replaced
by a unitary operator on the internal space, also known as \emph{coin}
space \cite{Aharonov:2001aa,Ambainis:2001aa}, determining the
evolution of the internal degree of freedom of the system.  This
model, known as coined quantum walk, has been proven to provide a
computational speedup over classical random walks for a class of
problems---such as some oracular problems, element distinctness
problem, {and the} triangle finding problem. The Grover's search algorithm can
also be implemented as a \QW
\cite{Childs:2003aa,Ambainis:2007aa,Magniez:2007aa,Farhi:2007aa,Santha:2008a,portugal2013quantum,Wong:2015a}.
The rigorous definition of \QW can be found in
Refs.~\cite{Ambainis:2001aa,Nayak:2000aa} for the one-dimensional
case, and in \cite{Aharonov:2001aa} for \QW{s} on graphs of any
dimension (see also \cite{Kempe:2003aa} for a complete review
including walks with continuous time evolution not considered in the
present context).

Aside from the interest in quantum algorithms, \QW{s} provide a fully quantum
model of evolution for a system with an internal degree of freedom. As such, \QW{s} have been considered as discrete quantum {simulators} 
 for
particle-physics. Interestingly, it has been proven that \QW{s} have the
capability of simulating free relativistic particle dynamics
\cite{Succi:1993aa,Bialynicki-Birula:1994ab,Meyer:1996aa,
  Strauch:2006aa,Yepez:2001ab,DAriano:2012af,Bisio:2015aa,
  Bisio:2013ab,Arrighi:2014aa,Arrighi:2013ab,Farrelly:2014ab,Farrelly:2014ac,Katori:2005aa,
PhysRevA.75.022322},
providing---in contrast with other discretisation schemes based on
finite-differences and which in general do not preserve the norm---a
local unitary model underlying relativistic dynamics. 

In the light of this success, in Ref.~\cite{DAriano:2014ae} the
authors propose a discrete theory for quantum field dynamics based on
finite dimensional quantum systems in interaction. Assuming the
locality, homogeneity, and unitarity of the interaction, it follows that
the systems must evolve according to a \QW. Moreover, the above
assumptions are very restrictive and the only \QW{s} admissible
on the cubic lattice in one, two, and three dimensions are proved to
recover the usual relativistic Weyl equation in the limit of
small wave-vectors. The massive case is obtained coupling two massless
\QW{s}, and in other works also the Maxwell equation \cite{Bisio2016} for
Bosonic fields is proved to be compatible with an elementary \QW model.
Finally, the Lorentz covariance, which is broken by the discreteness
of the walk, can be recovered as an approximated symmetry
\cite{bibeau2015doubly} in the relativistic limit. These results show
how \QW{s} not only provide a useful way of simulating relativistic free
evolution ,but also can be considered as a fundamental approach to
quantum field theory (see Refs.~\cite{Bisio2015-1,Bisio2015-2} for a
review).

Here we consider the Dirac \QW{s} derived in Ref.~\cite{DAriano:2014ae}
and present both an analytical and a numerical study of their
kinematics, recovering the characteristic traits of the usual Dirac
equation.

We show “smooth-states” peaked around some wave-vector eigenstate of
the \QW can be considered as particle states. We present an analytical
approximation of particle states evolution deriving a
wave-vector-dependent differential equation for the walk evolution.
Then we analyse in {detail}
 dynamical quantities {such} as the walk
position and velocity operators and study their evolution. An
intrinsic relativistic quantum processes of the Dirac field, denoted
\emph{Zitterbewegung}, first considered by Schr\"odinger
\cite{Schrodinger:1930aa} and corresponding to a jittering of the mean
position for a relativistic particle, is recovered from the \QW
evolution. The theoretical existence of the quivering motion has been
evidenced by numerical simulations of the Dirac equation and of
quantum field theory. While \emph{Zitterbewegung} oscillations cannot be
directly observed by current experimental techniques for a Dirac
electron since the amplitude should by very small (equal to the
Compton wavelength $\hbar/mc$ with $m$ the rest mass of the
relativistic particle, namely $\approx 10^{-12}$ m for an electron),
solid state and atomic physics provide physical hardware to simulate
the phenomenon
\cite{cannata1991effects,ferrari1990nonrelativistic,cannata1990dirac,lurie1970zitterbewegung,PhysRevLett.95.187203,PhysRevLett.99.076603,
  lamata2007dirac,gerritsma2010quantum,cserti2006unified,rusin2007transient,schliemann2005zitterbewegung,
  zawadzki2005zitterbewegung,zawadzki2010nature,geim2007rise,zawadzki2011zitterbewegung,vaishnav2008obserVing,PhysRevLett.100.113903},
and they have recently boosted a renewed interest in the Dirac equation
features.

%%%%%%%%%%%%%%%%%%%%%%%%%%%%%%%%%%%%%%%%%%
\section{Quantum Walks}

A quantum walk is a \emph{local} unitary evolution of a quantum system with
Hilbert space \mbox{$\H = \ell^2(V) \otimes \C^s$,} where $V$ is a countable
set and $\C^s$ is called \emph{coin space}---namely, the internal degree
of freedom of the walker, $s>0$ integer.  

A \QW on $\H$ is defined by assigning a mapping
$\mathcal{E}\colon V \times V \to \mathcal{M}_s(\C)$, such that
$\mathcal{E}(x,y) \coloneqq U_{y,x}$, which associates to each pair of
vertices a matrix, called \emph{transition matrix}, acting on the coin
space.  Then $\psi \colon \N \to \H$ is a solution of the QW
$(V, \mathcal{E})$ if it satisfies the following update rule for a
given initial condition $\psi(0) \in \H$:
\begin{equation} \label{eq:upd-rule}
	\psi(x,t+1) = \sum_{y\in V} \mathcal{E}(y,x) \psi(y,t), 
		\quad \forall\mkern1mu x \in V, \, \forall\mkern1mu t \in \N,
\end{equation}
where $\psi(x,t)\in\C^{s}$ is the \QW  wave-function. 
Since the \QW evolution is unitary, the transition matrices should satisfy the following conditions for all $x,y \in V$:
\begin{align}\label{eq:un-cond}
\sum_{z \in V} \EM(z,x)\EM(z,y)^\dagger = 
\sum_{z \in V} \EM(x,z)^\dagger\EM(y,z) =
\delta_{xy} I_s,
\end{align}
where $I_s$ is the identity on the coin space $\C^s$.
Such a \QW carries an
associated graph defined by the set of non-null transition matrices as
the directed graph $\Gamma = (V,E)$ with vertex set $V$ and edge set
$E \coloneqq \set{(x,y) \in V \times V | A_{y,x}
  \neq 0}$. The \emph{locality} condition amounts to requiring that for every vertex $x \in V$, the cardinality of its out-neighbourhood $\Neigh_x^{+}=\set{y\in V|U_{y,x}\neq 0}$, and in-neighbourhood  $\Neigh_x^{-}=\set{y\in V|U_{x,y}\neq 0}$, is uniformly bounded over $V$---namely $\card{\Neigh_x^{\pm}} < M < +\infty$ for every $x \in V$.

In Ref.~\cite{DAriano:2014ae} it has been shown that assuming the \QW \emph{homogeneous}
(the vertices of the graph cannot be distinguished by the walk
dynamics) the graph $\Gamma$ is actually a Cayley graph of a group
$G$.  Given a group $G$ and taking $S \subseteq G$, the Cayley graph
$\Cay{G}{S}$ of $G$ with connection set $S$ is defined as the coloured
directed graph $(G, S, E)$ with vertex set $G$, edge set
$E \coloneqq \set{(g,gh) | g \in G, \, h \in S}$ and colouring given
by $E \ni (g,g') \mapsto g^{-1}g' \in S$.  We will assume hereafter
that the connection set $S$ is a generating set for $G$---which
entails that the Cayley graph $\Cay{G}{S}$ is unilaterally
connected---and it is symmetric. Namely, $S = S^{-1}$.
The walk unitary operator corresponding to the update rule of \mbox{Equation (\ref{eq:upd-rule})} can be expressed in terms of the right-regular representation of $G$ on $\ell^2(G)$ defined as the map $G \ni g \mapsto T_g \in \Aut{\ell^2(G)}$ such that $T_g \ket{g'} = \ket{g'g^{-1}}$.
Assuming, by homogeneity, that we can choose the transition matrices independently of the vertex so that $U_{g,gh} \equiv U_{h}$ for every $g \in G$ and $h \in S$, we can write the walk operator $U \in \Aut{\H}$ as
\intentionalspace
\begin{equation} \label{eq:walk-op}
	U = \sum_{h \in S} T_h \otimes U_h.
\end{equation}

Now the unitarity conditions on $U$ translate into the following conditions of the transition \mbox{matrices $U_h$:}
\begin{align}\label{eq:unitarity}
\sum_{h\in S}U^\dag_h U_h=\sum_{h\in S}U_h U^\dag_h=I_s,\quad
\sum_{\shortstack{$\scriptstyle h,h'\in S$\\ $\scriptstyle
		h^{-1}h'=h''$}} U^\dag_h U_{h'}=\sum_{\shortstack{$\scriptstyle
		h,h'\in S$\\ $\scriptstyle h'h^{-1}=h''$}} U_{h'} U^\dag_{h}=0.
\end{align}

%%%%%%%%%%%%%%%%%%%%%%%%%%%%%%%%%%%%%%%%%%
\subsection{Fourier Representation of Abelian \QW{s}}

As pointed out by Ambainis \emph{et~al.} \cite{Ambainis:2001aa}, there
are two general ways to study the evolution of a \QW.  On the one hand,
one can exploit the algebraic properties of the walk transition
matrixes to obtain a path-sum solution, where the \QW transition
amplitude to a given site is expressed as a combinatorial sum over all
the paths leading to that site.  Regarding this approach, in
Ref. \cite{Ambainis:2001aa} the authors provided a solution for the
\emph{Hadamard Walk}, whereas Konno derived the solution for an
arbitrary coined \QW \cite{Konno:2002ab}.  Considering the application
of \QW{s} to the description of relativistic particles, also the Dirac \QW
in $1+1$-dimensions and the massless Dirac \QW in $2+1$-dimensions have
been analytically solved in position space
\cite{DAriano:2014ad,DAriano:2015aa}.  On the other hand, when a
\QW is defined on the Cayley graph of an Abelian group, the walk
dynamics can be studied in its Fourier representation, providing
analytical solutions and also approximate asymptotic solutions in the
long-time limit.

Let us now consider \QW{s} defined on Cayley graphs of free Abelian
groups. That is, $G \cong \Z^d$ with generating set $S$.  Adopting the
usual additive notation for the group operation on $\Z^d$, the
right-regular representation of $\Z^d$ is expressed as
\intentionalspace
\begin{equation}
	T_{\vec y} \ket{\vec x} = \ket{\vec x - \vec y}, 
		\qquad \forall\mkern1mu \vec{x}, \vec y \in \Z^d.
\end{equation}

Moreover it decomposes in one-dimensional irreducible representations, as can be easily seen 
from the fact that the translations $T_{\vec x}$ are diagonal on the plane waves
\intentionalspace
\begin{equation}
    \ket{\vec k} \coloneqq \frac{1}{\sqrt{\vol{\mathsf{B}}}}
        \sum_{\vec x \in \Z^d} e^{-\ims \vec k \cdot \vec x} \ket{\vec x}, \quad
    T_{\vec x} \ket{\vec k} = e^{-\ims \vec k \cdot \vec x} \ket{\vec k},
\end{equation}
\intentionalspace
where $\Br$ denotes the \emph{first Brillouin zone}, which depends on the specific Cayley graph $\Cay{\Z^d}{S}$ we \mbox{are considering.}

Therefore, we can write the walk operator of Equation (\ref{eq:walk-op}) in the direct integral decomposition
\intentionalspace
\begin{equation} \label{eq:walk-K}
    U = \int_{\Br}^{\oplus} \KetBra{\vec k}{\vec k} \otimes U_{\vec k} \dif{\vec k}, \qquad
    U_{\vec k} \coloneqq \sum_{\vec h \in S} e^{-\ims \vec k \cdot \vec h} U_{\vec h}.
\end{equation}
\intentionalspace

For each $\vec k$ we can diagonalise the matrix $U_{\vec k}$ obtaining
\intentionalspace
\begin{equation} \label{eq:eigen-walk}
    U_{\vec k} \ket{u_r(\vec k)} = e^{-\ims \omega_r(\bk)} \ket{u_r(\vec k)},
\end{equation}
\intentionalspace
where $\omega_r(\bk)$ is the \emph{dispersion relation} of the walk
and $\ket{u_r(\vec k)} \in \C^s$ is the eigenvector of $U_{\vec k}$
corresponding to the eigenvalue $e^{-\ims \omega_r(\bk)}$, with
$r = 1,\dots,s$. We notice that we have considered the representation
given by the factorized orthonormal basis $\ket{\vec{x}}\otimes \ket{r}$ for
the walk Hilbert space $\ell_2(\mathbb{Z}^d)\otimes \mathbb{C}^s$,
with $\{\ket{r}\}_{r=1}^{s}$ the canonical basis in
$\mathbb{C}^s$.

\section{The Dirac \QW in One, Two, and Three Space Dimensions}\label{s:dirac-qw}

In This Section, we present the Dirac \QW{s} in one, two, and three space
dimensions derived in Ref.~\cite{DAriano:2014ae}. We will see that in
the limit of small wave-vectors, the Dirac walks simulate the usual
Dirac equation evolution. We start from the simplest case of massless
Dirac \QW, also denoted Weyl \QW. The massive walk will be given by
coupling two Weyl \QW{s} with the coupling parameter interpreted \emph{a
posteriori} as the mass of the Dirac field.

\subsection{The Weyl Quantum Walk}\label{s:weyl}

In Ref.~\cite{DAriano:2014ae}, the authors derive the unique \QW{s} on
Cayley graphs of $\Z^d$ for $d = 1,2,3$ satisfying---besides locality
and unitarity---the assumptions of homogeneity and discrete isotropy,
and with minimal dimension $s$ of the coin space to have non-identical
evolution.  As first noticed by Meyer~\cite{Meyer:1996aa}, the only
solution for \emph{scalar} \QW{s} on Cayley graphs of free-Abelian
groups is the identical \QW; in order to have non-trivial dynamics, one
has to take at least $s=2$.

Let us start from Cayley graphs of $\Z^3$, the most relevant from the physical perspective.
It can be proved (see~\cite{DAriano:2014ae}) that only the body-centred cubic lattice (\BCC) allows one to define a \QW satisfying the above assumptions.
The \BCC lattice is the Cayley graph $\Cay{\Z^3}{S_+ \cup S_-}$ where $S_+ = \{\bh_1, \bh_2, \bh_3, \bh_4\}$ is the set of generators of the group and $S_-$ is the corresponding set of their inverses; a convenient choice for the generators is the following:
\intentionalspace
%
%\begin{equation} \label{eq:bcc-vecs}
%\vec{h}_{1} = \frac{1}{\sqrt{3}} \mvec{ 1\\ 1\\ 1}, \quad
%\vec{h}_{2} = \frac{1}{\sqrt{3}} \mvec{ 1\\-1\\-1}, \quad
%\vec{h}_{3} = \frac{1}{\sqrt{3}} \mvec{-1\\ 1\\-1}, \quad 
%\vec{h}_{4} = \frac{1}{\sqrt{3}} \mvec{-1\\-1\\ 1}.
%\end{equation}
\begin{equation} \label{eq:bcc-vecs}
\vec{h}_{1} = \mvec{ 1\\ 1\\ 1}, \quad
\vec{h}_{2} = \mvec{ 1\\-1\\-1}, \quad
\vec{h}_{3} = \mvec{-1\\ 1\\-1}, \quad 
\vec{h}_{4} = \mvec{-1\\-1\\ 1}.
\end{equation}
\intentionalspace

The first Brillouin zone $\Br$ of the BCC lattice is defined in Cartesian coordinates as $-\pi \leq k_i\pm k_j\leq \pi, \, i \neq j, \, i,j \in\{x,y,z\}$ and it is depicted in Figure \ref{fig:bcc-br}.

\begin{figure}
    \centering
    \includegraphics[width=0.3\textwidth]{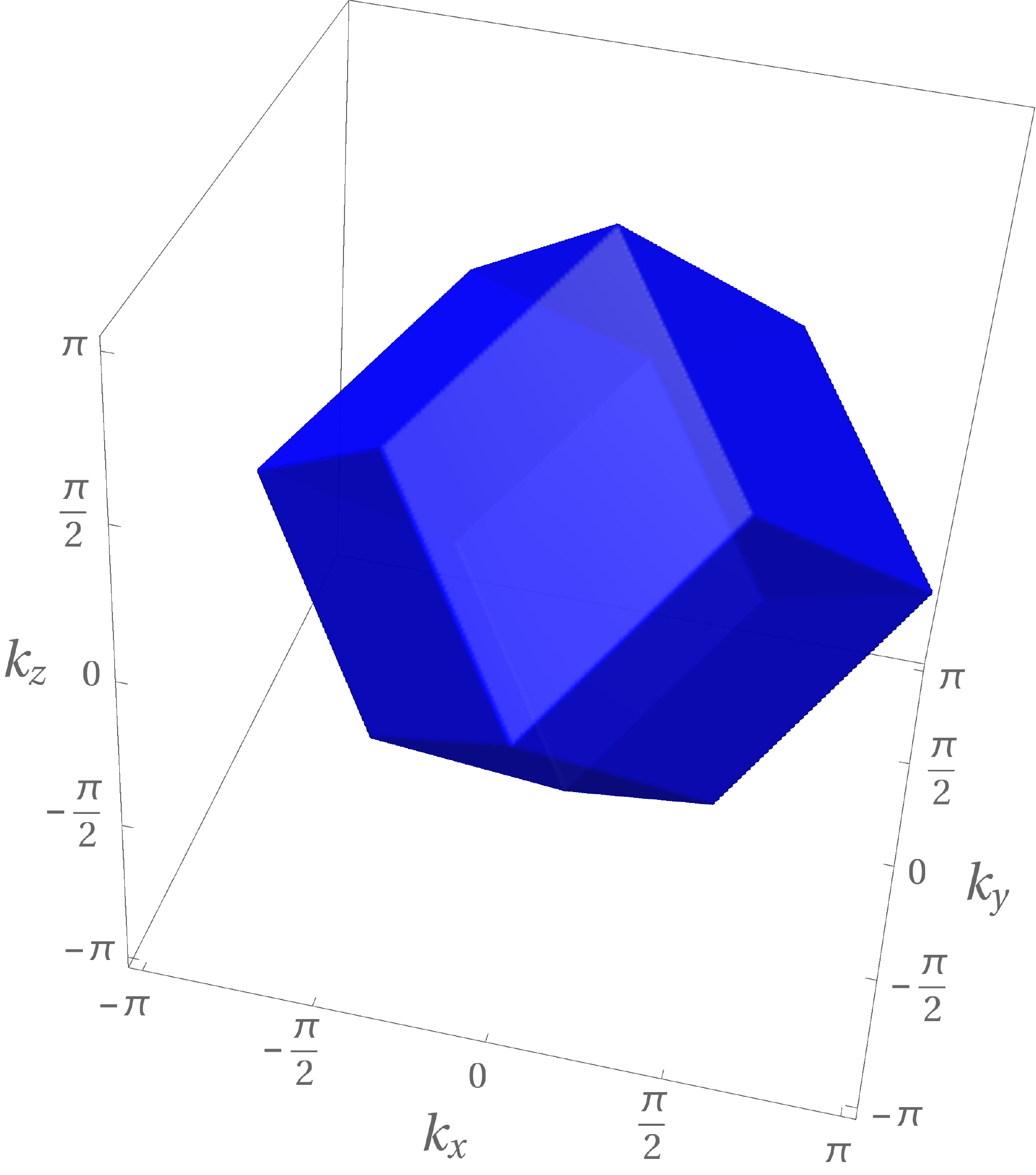}
    \caption{First Brillouin zone of the \BCC lattice (rhombic dodecahedron).}
    \label{fig:bcc-br}
\end{figure}

The unique solutions on the \BCC lattice can be summarised as:
\intentionalspace
\begin{equation}\label{eq:weyl3d}
\begin{aligned}
    &  U_{\bk} = u_\bk I - \ims \bsi \cdot \tilde\bn_\bk, \qquad
    && u_{\bk} \coloneqq c_x c_y c_z + s_x s_y s_z, \\
    &  \tilde\bn_{\bk} \coloneqq
        \mvec{s_x c_y c_z - c_x s_y s_z \\
              c_x s_y c_z + s_x c_y s_z \\
              c_x c_y s_z - s_x s_y c_z}, \qquad
    && c_i \coloneqq \cos k_i, \quad s_i \coloneqq \sin k_i.
\end{aligned}
\end{equation}
where $\mat{\sigma}$ is the vector with components given by the Pauli
matrices $\sigma_x$, $\sigma_y$ and $\sigma_z$. The walk matrix $U_\bk$
has spectrum $\{e^{-\ims\omega_\bk}, \, e^{\ims\omega_\bk}\}$ with
dispersion relation $\omega_\bk = \arccos u_\bk$ and group velocity
$\bv_\bk \coloneqq \nabla_\bk\omega_\bk$, representing the speed of a
wave-packet peaked around the central wave-vector $\bk$.

Let us consider now $d=2$; also in this case our assumptions single
out only one Cayley graph of $\Z^2$, the square lattice, involving two
generators $S_+ = \{\bh_1,\bh_2\}$, with $\bh_1 = (1,0)$ and
$\bh_2 = (0, 1)$; the first Brillouin zone $\Br$ in this case is given
by $-\pi \leq k_i \leq \pi, \, i \in\{x,y\}$, where $k_x = k_1 + k_2$
and $k_y = k_1 - k_2$.  The unitary matrix of the walk in Fourier
representation is given by:
\intentionalspace
\begin{equation}
\begin{aligned}\label{eq:weyl2d}
    &  U_{\bk}   = u_\bk I - \ims \bsi \cdot \tilde\bn_\bk, \qquad
    && u_{\bk}   \coloneqq c_x c_y, \\
    &  \tilde\bn_{\bk} \coloneqq
        \mvec{s_x c_y \\
              c_x s_y \\
              s_x s_y}, \qquad
    && c_i \coloneqq \cos k_i, \quad s_i \coloneqq \sin k_i,
\end{aligned}
\end{equation}
\intentionalspace
with dispersion relation $\omega_\bk = \arccos u_\bk$.

Finally, for $d=1$, the unique Cayley graph satisfying our requirements for $\Z$ is the lattice $\mathbb Z$ itself, considered as the
free Abelian group on one generator $S_+=\{h\}$.
From the unitarity conditions one gets the unique solution
\intentionalspace
\begin{equation}\label{eq:weyl1d}
\begin{aligned}
    U_k = u_k I - \ims \bsi \cdot \tilde\bn_k, \qquad
    u_{k} \coloneqq  \cos k, \quad
    \tilde\bn_{k} \coloneqq \mvec{0\\0\\\sin k},
\end{aligned}
\end{equation}
\intentionalspace
with dispersion relation $\omega_k = k$.

From \Cref{eq:weyl3d,eq:weyl2d,eq:weyl1d}, we see that the Weyl \QW in
dimension $d\leq 3$ is of the form
\intentionalspace
\begin{align} \label{eq:walk-gen}
    W_\bk = u_\bk I - \ims \bsi \cdot \tilde\bn_\bk,
\end{align}
\intentionalspace
for certain $u_\bk$ and $\bn_\bk$ with dispersion relation
\intentionalspace
\begin{align}
    \omega_\bk = \arccos{u_\bk}.
\end{align}
\intentionalspace

Now it is easy to show that the evolution of the walks
\Cref{eq:weyl3d,eq:weyl2d,eq:weyl1d} obeys Weyl's equation in the
limit of small wave-vectors, and thus we call them Weyl \QW{s}. Let us
introduce the interpolating Hamiltonian $\HIW$ defined in the
wave-vector space as the matrix such that $W_\bk = e^{-\ims \HIW}$ and
governing the continuous-time evolution, interpolating exactly the
discrete dynamics of the walk. \mbox{As one} can check, the interpolating
Hamiltonian is
\intentionalspace
\begin{align}\label{eq:weyl-interpolating}
    \HIW = \omega_\bk \bsi \cdot \bn_\bk, \qquad 
    \bn_\bk \coloneqq
        \frac{1}{\sin\omega_\bk} \tilde{\bn}_{\bk},
\end{align}
\intentionalspace
and by power expanding at the first order in $\bk$, one has
\intentionalspace
\begin{align}
    \HIW = \bsi \cdot \bk + \mathcal O(|\bk|^2)
\end{align}
\intentionalspace
whose first order term $\bsi \cdot \bk$ coincides with the usual Weyl
Hamiltonian in $d$ dimensions, with $d=1,2,3$, once the wave-vector $\bk$
is interpreted as the momentum.

It will be useful for the considerations of the following sections to
consider the eigenvectors of the \QW.  Since the structure of the
matrix is independent of the dimension, we will give here the general
expression of the eigenvectors.  Let us now rewrite the unitary matrix
$W_\bk$ as:
\intentionalspace
\begin{align}
    W_\bk = 
    \begin{pmatrix}
        z_\bk & -w^*_\bk \\
        w_\bk & z^*_\bk
    \end{pmatrix},
\end{align}
\intentionalspace
where $z_\bk$ and $w_\bk$ are related to the functions in Equation (\ref{eq:walk-gen}) by the equations $\Re(z_\bk) = u_\bk$ and $\tilde{\bn}_\bk = (-\Im(w_\bk), \Re(w_\bk), -\Im(z_\bk))$.
We can then solve the eigenvalue problem
\intentionalspace
\begin{align} \label{eq:eigen-W}
    W_\bk \ket{u_s^W(\bk)} = e^{-\ims s\omega_\bk} \ket{u_s^W(\bk)},
\end{align}
\intentionalspace
obtaining the following expression for the eigenvectors $\ket{u_s^W(\bk)}$, with $s=\pm$,
\intentionalspace
\begin{align} \label{eq:evecs-W}
    \ket{u_s^W(\bk)} = \frac{1}{\sqrt{2}}
    \mvec{\sqrt{1 - s v_\bk^W} \\ -s e^{\im\phi} \sqrt{1 + s v_\bk^W}}, \qquad
    v^W_\bk = \frac{\Im(z_\bk)}{\sqrt{1-u_\bk^2}},
\end{align}
\intentionalspace
where $\phi= \Arg w_\bk - \frac{\pi}{2}$.

\subsection{The Massive Case}

Now we present \QW{s} which manifest Dirac dynamics. We consider a walk
resulting from the local coupling of two Weyl \QW{s}.  One can show
\cite{DAriano:2014ae} that there is only one possible local coupling
of two Weyl \QW{s}, and that in the small wave-vector limit the resulting
walks approximate the Dirac's equation.  The unique local coupling of
Weyl's \QW{s}, modulo unitary conjugation, is of the form
\intentionalspace
\begin{equation}
    D_\bk =
    \begin{pmatrix}
        n W_\bk & \im m \\
        \im m   & n W_\bk^\dagger
    \end{pmatrix}, \quad
    n,m \in \R^+, \quad n^2 + m^2 = 1.
\end{equation}
\intentionalspace

We can provide a convenient expression of the walk in terms of the gamma matrices in spinorial~representation:
\intentionalspace
\begin{align}\label{eq:dirac}
    D_\bk = n u_\bk I - \im
        n \gamma_0 \bga \cdot \tilde\bn_\bk + \im m \gamma_0,
\end{align}
\intentionalspace
where $u_\bk$ and $\tilde\bn_\bk$ are those given previously for the Weyl's \QW{s}.
From Equation (\ref{eq:dirac}) we can see that the dispersion relation in this case is simply given by
\intentionalspace
\begin{equation}
    \omega_\bk = \arccos\paren*{\!\sqrt{1-m^2} \, u_\bk}.
\end{equation}

In this case, the interpolating Hamiltonian $\HI{D}$ has the form
\intentionalspace
\begin{align}\label{eq:dirac-interpolating}
  \HI{D} =
  \frac{\omega_\bk}{\sin\omega_\bk}
  (n \gamma_0 \bga \cdot \tilde\bn_\bk - m \gamma_0),
\end{align}
\intentionalspace
and to the first order in $\bk$ and $m$ one obtains the usual Dirac's
Hamiltonian
\intentionalspace
\begin{align}\label{eq:dirac-interp}
    \HI{D} = \gamma_0 \bga \cdot \bk + m \gamma_0 + \mathcal O(m^2) + \mathcal O(|\bk|^2).
\end{align}
\intentionalspace

We notice that the Dirac QW walk not only provides the usual Dirac
dispersion relation in the limit small wave-vectors, but also the
correct spinorial dynamics of the Dirac equation.

It is worth noticing that in dimension $d=1$, the Dirac \QW decouples
into two identical $s=2$ massive \QW{s}
\cite{DAriano:2014ae,Bisio:2015aa}, written explicitly as
\intentionalspace
\begin{align}
    D_k = 
    \begin{pmatrix}
        n e^{-\im k} & \im m \\
        \im m         & n e^{\im k}
    \end{pmatrix}, \qquad
    \omega_k = \arccos\paren*{n \cos k},
\end{align}
\intentionalspace
where $n,m \in \R^+$, $n^2 + m^2 = 1$.  In one space dimensions,
similar QW have been studied in the literature. For example in
Refs.~\cite{Strauch:2006aa,Katori:2005aa}, the authors consider the
relation between arbitrary coined QWs and relativistic dynamics.

For the massive \QW{s}, the eigenvalue equation takes the form:
\intentionalspace
\begin{align} \label{eq:eigen-D}
    D_\bk \ket{u_{s,p}^D(\bk)} = e^{-\ims s\omega_\bk} \ket{u_{s,p}^D(\bk)},
\end{align}
\intentionalspace
and the four eigenvectors $\ket{u_{s,p}^D(\bk)}$, with $s,p = \pm$ can be written as
\intentionalspace
\begin{align} \label{eq:evecs-D}
    \ket{u_{s,p}^D(\bk)} = \frac{1}{2}
    \mvec{\sqrt{(1 - p v_\bk^W)(1 + sp v_\bk^D)} \\
          -p e^{\im\phi} \sqrt{(1 + p v_\bk^W)(1 + sp v_\bk^D)} \\
          -s \sqrt{(1 - p v_\bk^W)(1 - sp v_\bk^D)} \\
          sp e^{\im\phi} \sqrt{(1 + p v_\bk^W)(1 - sp v_\bk^D)}}, \qquad
    v^D_\bk = \frac{n \sqrt{1-u_\bk^2}}{\sqrt{1-n^2 u_\bk^2}},
\end{align}
\intentionalspace
with $\phi$, $u_\bk$, and $v_\bk^W$ defined for the corresponding massless \QW of Equations (\ref{eq:eigen-W}) and (\ref{eq:evecs-W}).

\section{Numerical Simulation of the Weyl and Dirac \QW{s}}\label{s:numerical}

In order to evaluate numerically the evolution of \QW{s}, one can adopt two different approaches. 
On the one hand, one can exploit the update rule in position space given by Equation \eqref{eq:upd-rule}, which is straightforward to implement numerically. 
This approach, however, is not very efficient if we only want to know the evolved state at some specific time $t$, since it would require $t$ successive updates of the state.
On the other hand, the Fourier representation of the walk allows one to {directly compute} %Was compute directly
 the evolution at a specific time---the complexity of the computation being that of the Fourier Transform, which can be efficiently implemented via a Fast Fourier Transform (\FFT) algorithm such as the Cooley--Tukey \FFT algorithm \cite{Cooley:1965a}.

Recalling the general expressions in Equations (\ref{eq:walk-K}) and (\ref{eq:eigen-walk}), the evolution of a state $\ket{\psi(0)} \in \H=\ell^2(\Z^d)\otimes\C^s$ is given by the subsequent application of the walk unitary $\ket{\psi(t)} = U^t \ket{\psi(0)}$.
Therefore, the state at time $t$ can be expressed in terms of its representation in Fourier space as the Fourier~Transform 
\intentionalspace
\begin{equation}
    \ket{\psi(\bx,t)} = \frac{1}{(2\pi)^{d/2}}
    \sum_{r=1}^s \int_{\Br} e^{-\ims \bk\cdot\bx} e^{-\ims \omega_r(\vec k) t} 
        \hat{\psi}_r(\vec k) \ket{u_r(\vec k)} \dif{\vec k},\qquad  \ket{\psi(\bx,t)}\in\C^{s},\; \bx\in\Z^d,
\end{equation}
\intentionalspace
where
\(
    \hat{\psi}_r(\vec k) = \sum_{\vec x \in G} 
        \braket{u_r(\vec k) | \psi (\bx,0)} e^{\ims \vec k \cdot \vec x}
\) is the $r$-component in the eigenbasis of the walk of the discrete-time Fourier transform of $\psi(\bx,0)$.
The notation of the eigenbasis refers here to that of \mbox{Equation (\ref{eq:eigen-walk}),} where $r$ runs over $\{1,\dots,s\}$ and $s$ is the dimension of the coin.
Now, the numerical data used to represent the state in Fourier space constitute a discrete sampling of it, say at frequencies $\frac{2\pi}{N_i} k_i$ with $k_i = -\floor{N_i/2},\dots,\ceil{N_i/2} - 1$ and $N_i$ the total number of samples in dimension $i$.
Implementing periodic boundary conditions, this amounts to {taking} %Was "take"
 samples in direct space over a finite region, extending the data periodically to the whole lattice.

Let us consider now the simple cubic lattice of $\Z^d$.
Let us consider a restriction $f\colon \Z^d \to \C$ of $\phi \in \ell^2(\Z^d)$ to a finite region $\mathcal{N} = \set{\v{m} \in \Z^d | 0 \leq m_i < N_i, \, i = 1,\dots,d}$, with $N_i \in \N$, such that $f|_{\mathcal{N}} = \phi|_{\mathcal{N}}$.
Then, $f|_{\mathcal{N}}$ is periodically extended to $\Z^d$; namely 
$f_{\bn + \v{N} \br} = f_\bn, \, \forall\mkern2mu \bn,\br \in \Z^d$, with periodicity matrix given by $\v{N} = \mathrm{diag}(N_1,\dots,N_d)$.
The Fourier Transform $\mathcal{F}$ of the sequence $f_\bn$ coincides with the Discrete Fourier Transform (\DFT) defined as 
\intentionalspace
\begin{equation} \label{eq:dft}
    \hat f_\bk = \FF(f)(\bk) \coloneqq \frac{1}{\sqrt{N}} 
        \sum_{\bn \in \mathcal{N}} 
        f_\bn e^{-2\pi\ims \v p^\TR \v{N}^{-1} \bn} 
        e^{2\pi\ims \bk^\TR \v{N}^{-1} \bn}, \quad
        \bk \in \mathcal{N},
\end{equation}
\intentionalspace
where $N = \card{\mathcal{N}} = \mathrm{det}(\v{N})$ and 
$p_i = \floor*{\frac{N_i}{2}}$.
The inversion formula is then given by:
\intentionalspace
\begin{equation} \label{eq:idft}
    f_\bn = \FF^{-1}(\hat f)(\bn) = e^{2\pi\ims \v{p}^\TR \v{N}^{-1} \bn} \frac{1}{\sqrt{N}} 
        \sum_{\bk \in \mathcal{N}} 
        \hat f_{\bk} e^{-2\pi\ims \bk^\TR \v{N}^{-1} \bn}.
\end{equation}
\intentionalspace

Here we have chosen the set of Fourier indices $\bk$ so that the frequencies actually computed lie in the interval $[-\pi,\pi]$.

For the Dirac \QW in $3+1$-dimensions, we have to consider instead the \BCC lattice.
One can show~\cite{Alim:2009aa} that it is possible to reduce the \DFT on the \BCC lattice to two rectangular \DFT{s}, allowing {implementation of} 
 the \DFT via usual rectangular \FFT algorithms.
We can describe the \BCC lattice choosing as vertex set $G = 2\Z^3 \cup (2\Z^3+\v{t})$, where $\v{t} = (1,1,1)$.
A suitable truncation of a sequence $\phi_\bn$, $\bn \in G$ to a function $f_\bn$ defined on a finite set $B \subset G$ can be obtained choosing the fundamental region $B = 2\mathcal{N} \cup (2\mathcal{N}+\v{t})$ and periodically extending it to $G$.
The original sequence $f_\bn$ can be further split into two subsequences on the even and odd indices $f^0_\bn = f_{2\bn}$ and $f^1_{\bn} = f_{2\bn + \v{t}}$, for all $\bn \in \mathcal{N}$.
As a consequence, these two sequences $f^0_\bn$ and $f^1_\bn$ are periodic with periodicity matrix $\v{N}$:
$f^j_{\bn+\v{N}\v{r}} = f^j_{\bn}$ for all $\bn$ and $\v{r}$ in $\Z^3$ and $j=0,1$.
The Fourier Transform of $f_\bn$ is defined as usual as
\intentionalspace
\begin{align}
    \hat f_\bk = \mathcal{F}(f)(\bk) \coloneqq 
        \frac{1}{2\sqrt{\card{B}}} 
        \sum_{\bn \in B} f_\bn e^{2\pi\ims \bk^\TR (2\v{N})^{-1} \bn},
        \qquad
        \forall\mkern1mu\bk \in \mathcal{K},
\end{align}
\intentionalspace
where the set of Fourier indices can be chosen as 
$\mathcal{K} = \Set{\bk \in \Z^3 | -N_i \leq k_i < N_i, \, i=1,2,3}$.
\mbox{As shown} in Ref.~\cite{Alim:2009aa}, one can exploit the geometry of the \BCC lattice to reduce the \DFT $\hat f_\bk$ with $\bk \in \mathcal{K}$ to two functions $\hat f^0_\bk$ and $\hat f^1_\bk$ with $\bk$ restricted now to $\mathcal{N}$.
This allows for the computation of the \DFT in terms of the usual rectangular \DFT{s}:
\intentionalspace
\begin{align} \label{eq:bcc-dft}
	\hat f^0_\bk & = \frac{1}{\sqrt{2}}
    \sparen*{\FF(f^0)(\bk) - a_\bk \FF(f^1)(\bk)}, \\
    \hat f^1_\bk & = \frac{1}{\sqrt{2}}
    \sparen*{\FF(f^0)(\bk) + a_\bk \FF(f^1)(\bk)},
\end{align}
\intentionalspace
with $\bk \in \mathcal{N}$ and $a_\bk = e^{\pi\ims \bk^\TR\v{N}^{-1}\v{t}}$. 
Finally, from the two sequences $\hat f^0_\bk$ and $\hat f^1_\bk$, we can write the inversion formulae for $f^0_\bn$ and $f^1_\bn$ as:
\intentionalspace
\begin{align} \label{eq:bcc-idft}
    f^0_\bn & = \frac{1}{\sqrt{2}} \FF^{-1}(\hat f^0 + \hat f^1)(\bn), \\
    f^1_\bn & = \frac{1}{\sqrt{2}}
    \FF^{-1}\paren{a^*(\hat f^1 - \hat f^0)}(\bn).
\end{align}

%%%%%%%%%%%%%%%%%%%%%%%%%%%%%%%%%%%%%%%%%%

\section{Kinematics of the Dirac \QW}

Here we study the kinematics of the Dirac \QW presented in Section
\ref{s:dirac-qw}. We show that there exists a class of states whose
evolution resembles the evolution of a particle with a given
wave-vector. \mbox{Their evolution} can be described by an approximated
differential equation with coefficients depending on the particle
wave-vector. We observe that the positive and negative frequency
eigenstates of the walk correspond to Dirac particle and antiparticle
states.

Finally, we consider the position operator for the Dirac \QW and find
that the mean position of states having both positive and negative
frequency components present the typical jittering phenomenon---denoted
\emph{Zitterbewegung}---of relativistic particles. The \emph{Zitterbewegung} was
first discovered by Schr\"odinger in 1930
\cite{schrodinger:1930aa}, who pointed out that in the Dirac
equation for free relativistic electrons, the velocity operator does
not commute with the Dirac Hamiltonian. As a consequence, the evolution
of the position operator shows---in addition to the classical motion
proportional to the group velocity---a fast periodic oscillation with
frequency $2mc^2$ and amplitude equal to the Compton wavelength
$\hbar/mc$, with $m$ the rest mass of the relativistic particle. This
oscillating motion is due~\cite{huang1952zitterbewegung} to the
interference of states corresponding to the positive and negative
energies firstly appeared as solutions to the Dirac equation. The
trembling is also shown to disappear with time
\cite{lock1979zitterbewegung} for a wave-packet particle state.  The
same phenomenology is recovered in the Dirac \QW scenario that also
presents solutions having positive and negative \mbox{frequency eigenvalues.}

\subsection{Approximated Dispersive Differential Equation}

We denote a quantum state of the walker a particle state if it is
localized in a region of the lattice at a given instant of time and if
the walk evolution preserves its localization. Accordingly, we take the
following of particle state as a state that is narrow-banded in the
wave-vector space.

\begin{Definition}[Particle-state]\label{d:particle-state}
A particle-state $\ket{\psi}$ for the Dirac \QW is a wave-packet smoothly peaked
around some eigenvector $\ket{u(\v{k}')}$ of the walk. Namely, for a given $\bk'$
\begin{equation}
\label{eq:smoothstate}
\ket{\psi} = \frac{1}{(2\pi)^{d/2}} \int_{\mathsf{B}}\dif{\v{k}}\,
g_{\v{k}'}(\v{k}) \ket{\v{k}}\ket{u(\v{k})},
\end{equation}
where $g_{\v{k}'}\in C^\infty_0[\mathsf{B}]$ is a smooth function satisfying the bound
\begin{equation}\label{eq:smoothstate2}
\begin{aligned}
&\frac{1}{(2\pi)^{d}}\int_{\mathsf{B}_{\v{k}'} (\sigma_x,\sigma_y,\sigma_z)}\dif{\v{k}}\,
| g_{\v{k}'} (\v{k}) |^2 \geq 1-\epsilon,\qquad \epsilon
>0,\quad\sigma_i >0, \,i=x,y,z,\\
&{\mathsf{B}_{\v{k}'} (\sigma_x,\sigma_y,\sigma_z)}=\set{\v{k}\in
\mathsf{B} | |k_i-k'_i|\leq \sigma_i}.
\end{aligned}
\end{equation}
\end{Definition}

In the next Proposition, we derive a dispersive differential
equation governing the evolution for particle-states and which makes
clear their particle behaviour. It will be convenient to work with the
continuous time $t$, interpolating exactly the discrete walk
evolution $U^t$. Accordingly, we consider $\v{x}$, $t$ to be
real-valued continuous variable by extending the Fourier transform
\begin{align}
\ket{\psi(\v{x},t)}=\frac{1}{(2\pi)^{d}}\int_\mathsf{B}\dif{\v{k}}\,
e^{i\v{k}\cdot\v{x}} U^t_{\v{k}}\ket{\psi(\v{k},0)}=\frac{1}{(2\pi)^{d}}\int_\mathsf{B}\dif{\v{k}}\,
e^{i\v{k}\cdot \v{x}} e^{-iH(\bk)t}\ket{\psi(\v{k},0)},
\end{align}
to real $\v{x}$, $t$. Since the walk is band-limited in momenta
$\v{k}\in \mathsf{B}$, then the continuous function $\psi(\v{x},t)$ is
completely defined by its value on the discrete points $(\v{x},t)$ of
the walk causal network (the sampling of a band-limited function
is stated in the Nyquist-Shannon {\em Sampling Theorem}).  However,
all numerical results will be given only for the discrete $t$, namely
for repeated applications of the walk unitary operator and for
discrete lattice sites $\v{x}$.

\begin{Proposition}[Dispersive differential equation]\label{p:semi-classical-states}
  Consider the evolution of the Dirac \QW of Section \ref{s:dirac-qw} on a
  particle-state as in Definition \ref{d:particle-state}. Then for any positive
  integer $n$, the state at time $t$ is given by
\begin{align}\label{eq:l-solution}
  \ket{\psi(\v{x},t)}=e^{i(\v{k}'\cdot \v{x}-\omega_{\v{k}'}
  t)}\Ket{\tilde{\phi}(\v{x},t)}-\epsilon -\gamma\Sigma^{n+1} t - {\cal
  O}(\Sigma^{n+3})t,
\end{align}
where $\Ket{\tilde{\phi}(\v{x},t)}$ is solution of the following
differential equation
\begin{align}\label{eq:l-diff-eq}
  i\partial_t
  \Ket{\tilde{\phi}(\v{x},t)}=\sum_{\alpha_x+\alpha_y+\alpha_z=n}\frac{(-i)^{n}\omega_{\v{k}'}^{(\alpha_x,\alpha_y,\alpha_z)}}{\alpha_x
  !\alpha_y
  !\alpha_z
  !}\frac{\partial^{n}}{\partial x^{\alpha_x}\partial y^{\alpha_y}\partial z^{\alpha_z}}
  \Ket{\tilde{\phi}(\v{x},t)},\\
\omega_{\v{k}'}^{(\alpha_x,\alpha_y,\alpha_z)}:=\left.\frac{\partial^{n}\omega_{\v{k}}}{\partial k_x^{\alpha_x} \partial k_y^{\alpha_y} \partial k^{\alpha_z}_z}\right|_{\v{k}=\v{k}'}
  \end{align}
and
\begin{align}
\gamma= (n+1)\sum_{\alpha_x+\alpha_y+\alpha_z=n+1}\frac{ |\omega_{\v{k}'}^{(\alpha_x,\alpha_y,\alpha_z)}|}{(2
  \pi)^d \alpha_x!\alpha_y!\alpha_z!} \int_{\mathsf{B}_{\v{k}'}(\sigma_x,\sigma_y,\sigma_z)}
  \dif{\v{k}}\,  |g_{\v{k}'}(\v{k})|^2,\qquad \Sigma=\mathrm{max}(\sigma_x,\sigma_y,\sigma_z).
\end{align}
\end{Proposition}

\begin{proof}
First we notice that at time $t$ the particle-state in the
momentum representation is simply
$\ket{\psi(\v{k},t)}=e^{-i\omega_{\v{k}}
  t}\ket{\psi(\bk,0)}=e^{-i\omega_{\v{k}} t}g_{\v{k}'}(\v{k})
\ket{u(\v{k})}$, while in the position representation it is
\begin{align}\label{eq:smoothstatetimet-psi}
&\ket{\psi (\v{x},t) } \coloneqq e^{i(\v{k}' \v x -\omega_r(\v{k}') t) } \ket{{\phi}(\v{x},t)}, \\
&\ket{{\phi}(\v{x},t)} \coloneqq \frac{1}{(2\pi)^{d}}\int_{\mathsf{B}}\dif{\v{k}}\, e^{i(\v{K}\cdot
  \v{x}-\Omega_{\v{k}}t) }g_{\v{k}'}(\v{k})
   \ket{u(\v{k})},\\ \nonumber
&\v{K} = \v{k}-\v{k}',\\ \nonumber 
&\Omega_{\v{k}} = \omega_{\v{k}}-\omega_{\v{k}'}.
\end{align}

Now we take the time derivative of $\ket{\phi(\v{x},t)}$ and
expand $\Omega$ \textit{vs.} $\v{k}$ in $\v{k}'$. The coefficients of the
expansion can be regarded as derivatives with respect to the space
coordinates and taken out of the integral (dominated derivative
theorem), leading to the following dispersive differential equation:

\begin{align}\label{eq:diff-expansion}
  i\partial_t
  \Ket{\phi(\v x,t)}=\sum_{|\alpha|=1}^\infty\frac{(-i)^\alpha \omega_{\v{k}'}^{(\alpha)}}{\alpha !}\frac{\partial^{|\alpha|}}{\partial  \v{x}^\alpha}
  \Ket{\phi(\v{x},t)},\qquad
  \omega_{\v{k}'}^{(\alpha)}=\left.\frac{\partial^{|\alpha|}\omega(\v{k})}{\partial\v{k}^\alpha}\right|_{\v{k}=\v{k}'},
\end{align}
where $\alpha=(\alpha_x,\alpha_y,\alpha_z)$ is a multiindex and
$|\alpha|=\alpha_x+\alpha_y+\alpha_z$. If we truncate the above
expansion at the $n$th order and denote by
$\Ket{\tilde{\phi}(\v{x},t)}$ the solution of the corresponding
truncated differential equation, with the identification of the
initial condition
$\Ket{\tilde{\phi}(\v{x},0)} = \ket{{\phi}(\v{x},0)}$, we get the
approximate state Equation \eqref{eq:smoothstatetimet-psi} at time $t$
\begin{align}\label{eq:approxstate0}
 \Ket{\tilde{\psi}(\v{x},t)}= e^{i(\v{k}' \cdot\v{x} - \omega_{\v{k}'}t)} \Ket{\tilde{\phi}(\v{x},t)}.
\end{align}

Using the definition of particle-state in Definition
\ref{d:particle-state}, one can compute the accuracy of the
approximation Equation \eqref{eq:approxstate0} in terms of the parameters
$\sigma_x,\sigma_y,\sigma_z$, and $\epsilon$, evaluating the overlap
between the states Equations \eqref{eq:smoothstatetimet-psi} and
\eqref{eq:approxstate0}. That is,
\begin{equation}
\begin{small}
\begin{aligned}
  |\braket{\tilde{\psi}(t)|\psi(t)}| &= \left|
    \frac{1}{(2\pi)^d}\int_{\mathsf{B}}\dif{\v{k}} \,
    e^{-it(\sum_{|\alpha|=n+1} \frac{\omega_{\v{k}'}^{(\alpha)}}{\alpha!} (\v{k}-\v{k}')^\alpha
      +\sum_{|\alpha|=n+2}{\cal O}((\v{k}-\v{k}')^\alpha))
    }|g_{\v{k}'}(\v{k})|^2
  \right| \\
  &\geq
  \left|\frac{1}{(2\pi)^d}\int_{\mathsf{B}_{\v{k}'}(\sigma_x,\sigma_y,\sigma_z)}
    \dif{\v{k}}\, e^{-it(\sum_{|\alpha|=n+1}
      \frac{\omega_{\v{k}'}^{(\alpha)}}{\alpha!} (\v{k}-\v{k}')^\alpha +\sum_{|\alpha|=n+2}{\cal
        O}((\v{k}-\v{k}')^\alpha )) }|g_{\v{k}'}(\v k)|^2 \right| 
  \\
  &\qquad \qquad \qquad -\left|\frac{1}{(2\pi)^d}
    \int_{\mathsf{B}\setminus\mathsf{B}_{\v{k}'} (\sigma_x,\sigma_y,\sigma_z)}
    \dif{\v{k}}\, e^{-it(\sum_{|\alpha|=n+1}
      \frac{\omega_{\v{k}'}^{(\alpha)}}{\alpha!} (\v{k}-\v{k}')^\alpha +\sum_{|\alpha|=n+2}{\cal
        O}((\v{k}-\v{k}')^\alpha )) }|g_{\v{k}'}(\v{k})|^2 \right|
  \\
  &\geq\left| 1 - i (n+1) t \sum_{|\alpha|=n+1}\frac{ \omega_{\v{k}'}^{({\alpha})} \Sigma^{n+1}}{(2
      \pi)^d\alpha!} \int_{\mathsf{B}_{\v{k}'} (\sigma_x,\sigma_y,\sigma_z)} \dif{\v{k}}\, |g_{\v{k}'}(\v{k})|^2 -
    {\cal O}(\Sigma^{n+3})t \right|- \epsilon \\
&\geq 1 - \epsilon
  -\gamma\Sigma^{n+1} t - {\cal O}(\Sigma^{n+3})t\label{eq:overlap}
\end{aligned}
\end{small}
\end{equation}
where
\begin{align}\nonumber
\gamma= (n+1)\sum_{|\alpha|=n+1}\frac{ |\omega_{\v{k}'}^{(\alpha)}|}{(2
  \pi)^d\alpha!} \int_{\mathsf{B}_{\v{k}'}(\sigma_x,\sigma_y,\sigma_z)}
  \dif{\v{k}}\,  |g_{\v{k}'}(\v{k})|^2,\qquad \Sigma=\mathrm{max}(\sigma_x,\sigma_y,\sigma_z).
\end{align}

Therefore, the exact state $\ket{\psi(\v{x},t)}$ at time $t$ can be
approximated  by Equation (\ref{eq:approxstate0}) with the accuracy given by the
overlap in Equation (\ref{eq:overlap}). 
\end{proof}

This approximation fails to be accurate for a sufficiently large value of $t$. More precisely, if we require the overlap to satisfy $|\braket{\tilde{\psi}(t)|\psi(t)}| > 1-\delta$, for some $\delta > 0$, then for $t > \frac{\delta - \epsilon}{\gamma \Sigma^{n+1}}$ the approximated solution can deviate significantly from that of the \QW.
A typical application of the above proposition is the second order
approximation of the state evolution. In that case,
Equation~(\eqref{eq:l-diff-eq}) gives
\begin{equation}\label{eq:diff-eq-2}
\begin{aligned}
&  i\partial_t
 \Ket{\tilde{\phi}(\v{x},t)}=\left[-i\v{v}_{\v{k}'}\cdot\nabla-\frac{1}{2}\nabla^{T}\cdot\v{D}_{\v{k}'}\cdot\nabla\right]
  \Ket{\tilde{\phi}(\v{x},t)},\\
&\v{v}_{\v{k}}=\nabla_\bk\omega_{\v{k}},\qquad
\v{D}_{\v{k}}=\nabla_{\v{k}}\nabla_{\v{k}}\omega_{\v{k}},
  \end{aligned}
\end{equation}
where $\v{v}_{\v{k}}$ and $\v{D}_{\v{k}}$ are respectively to drift
vector and to the diffusion tensor for the particle state. Accordingly,
the state will translate with group velocity given by the drift vector
and its distribution in space will spread as described by the
diffusion tensor. 

In Figure~\ref{fig:gauss3d}, we show the numerical evolution (see Section
\ref{s:numerical}) of a Gaussian particle-state in $3+1$~dimensions.

\begin{figure}[H]
    \centering
    \includegraphics[width=0.35\textwidth]{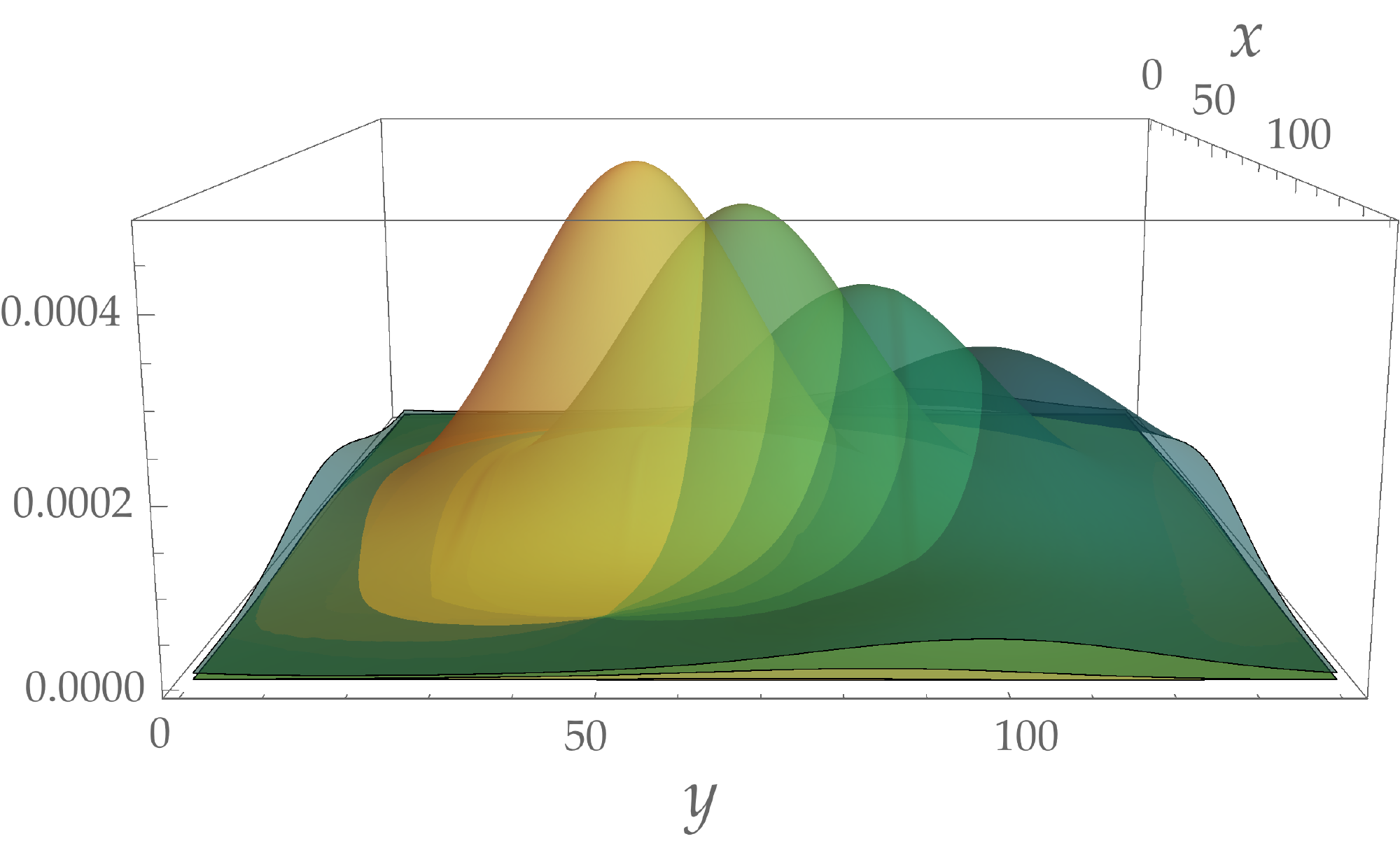}
    \caption{Evolution for $t=150$ time-steps of a particle state as in \Cref{d:particle-state} for the Dirac \QW in $3+1$ dimensions with only positive energy components (the marginal along the $z$ axis is represented). Here the state is Gaussian with parameters: mass $m=0.02$, mean wave-vector $\bk' = (0,0.01,0)$, width $\sigma_i = \sigma = 32^{-1}$ for $i = x,y,z$. Evolution in time is shown as a colour gradient from light to dark; one can notice the spreading of the wave-packet as time increases.
    Probability distribution shown at times $t=0,50,100,150$.}\label{fig:gauss3d}
\end{figure}

\subsection{The Evolution of the \QW Position Operator}

Up to now we have considered only smooth-states (see Definition
\ref{d:particle-state}) whose walk evolution is well described by the
approximate differential equation derived in Proposition
\ref{p:semi-classical-states}. On the other hand, in the \QW
framework we are allowed to consider states very far apart from the
smooth ones, and in the limit one can also consider perfectly localized
states as $\ket{\psi}=\ket{\v{x}}\ket{\zeta}$ with $\v{x}\in\Z^d$ and
$\ket{\zeta}=\sum_r c_r\ket{r}\in\C^4$, $\sum_r|c_r|^2=1$,
where $\{\ket{r}\}_{r=1}^4$ denotes the $\C^4$ basis
corresponding to the Dirac field representation in Equation~\eqref{eq:dirac}.
Since these states involve large momentum components, their evolution according to the \QW dynamics will be very different from the one given by the Dirac equation. 
We can say that the \QW determines different regimes with respect to a given reference scale at which the evolution deviates from the relativistic regime given by the Dirac equation \cite{bibeau2015doubly,Bisio:2015ac}.
However, in a \QW context, the study of such states can give essential information regarding the dynamical properties of these models \cite{Aharonov:2001aa,Ambainis:2001aa,Nayak:2000aa,Kempe:2003aa}.
One can see in Figures \ref{fig:loc-3d-3} and \ref{fig:loc-3d-proj} the numerical evolution (see Section \ref{s:numerical}) of a perfectly localised state, according to the Dirac \QW in $3+1$-dimensions.

\begin{figure}[H]
	\centering
	\begin{subfigure}[c]{0.25\textwidth}
		\includegraphics[width=\textwidth]{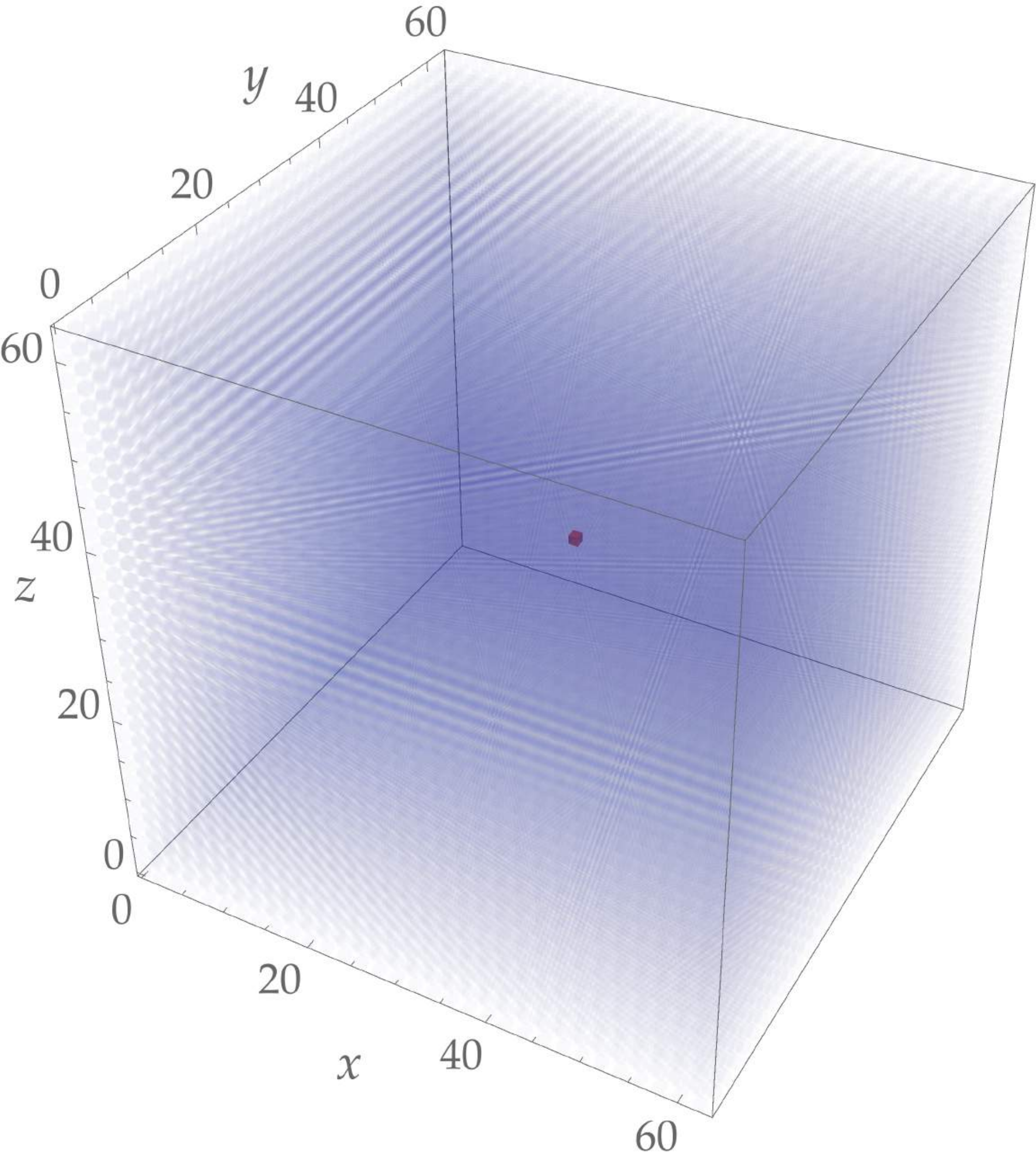}
	\end{subfigure}
	\quad
	\begin{subfigure}[c]{0.25\textwidth}
		\includegraphics[width=\textwidth]{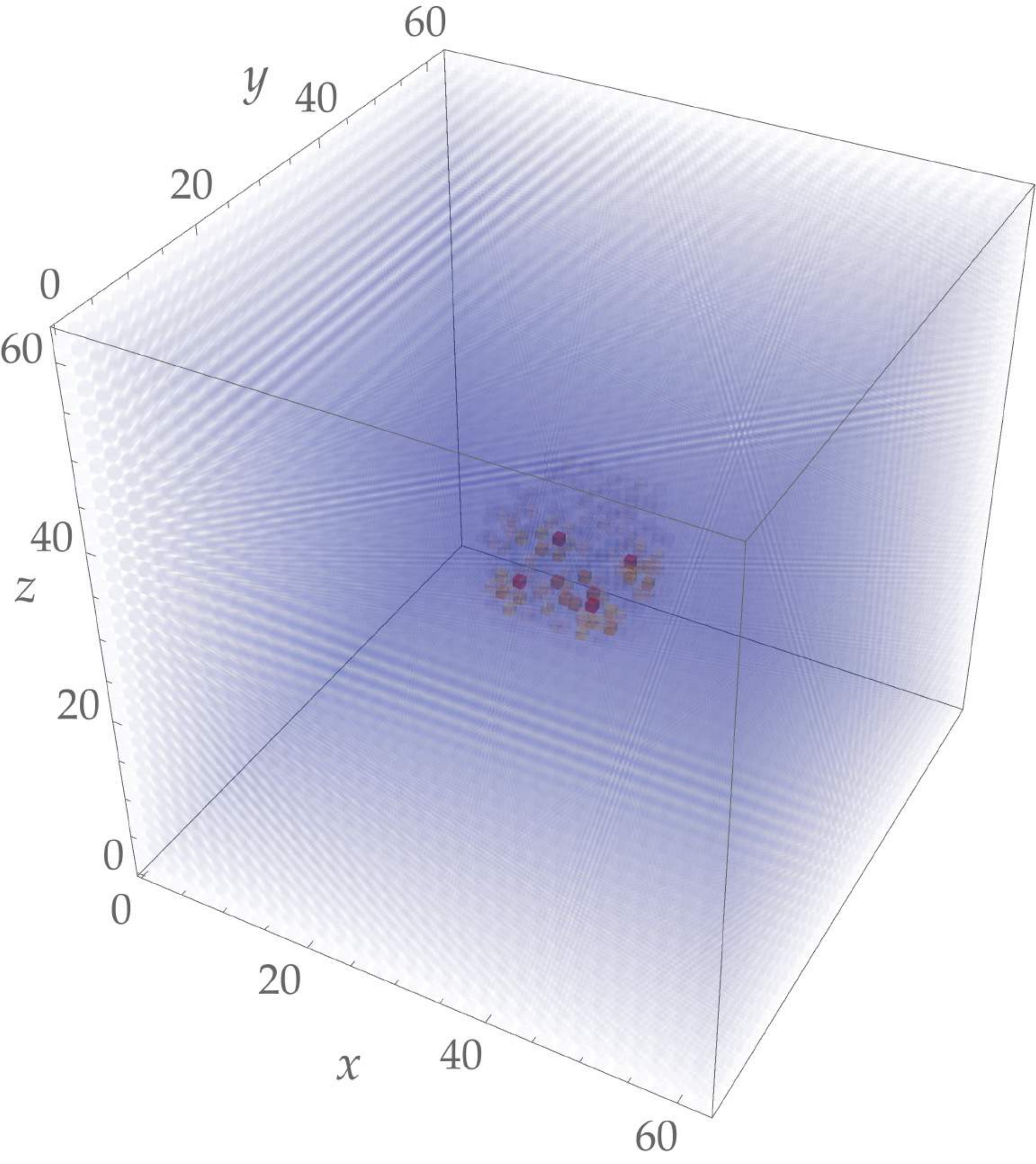}
	\end{subfigure}
	\quad
	\begin{subfigure}[c]{0.25\textwidth}
		\includegraphics[width=\textwidth]{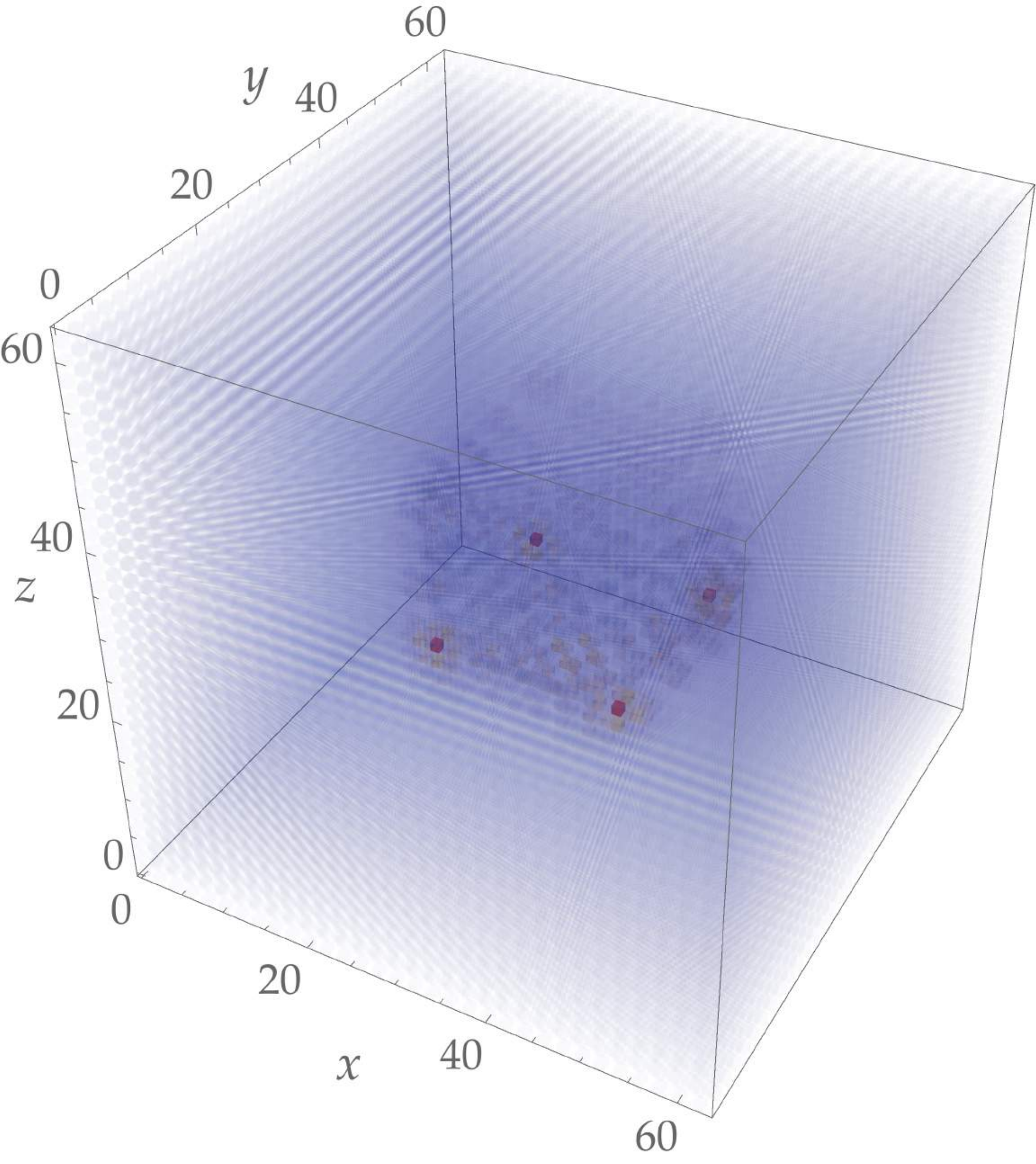}
	\end{subfigure}
	\caption{Evolution of a perfectly localised state for the Dirac \QW in $3+1$-dimensions. The figures show the probability distribution at times $t=0,8,16$, from left to right. In this case, the mass parameter is $m=0.03$ and the spinor in the canonical basis is $(1,0,0,0)$.}
	\label{fig:loc-3d-3}
\end{figure}

\begin{figure}[H]
	\centering
	\begin{subfigure}[c]{0.25\textwidth}
		\includegraphics[width=\textwidth]{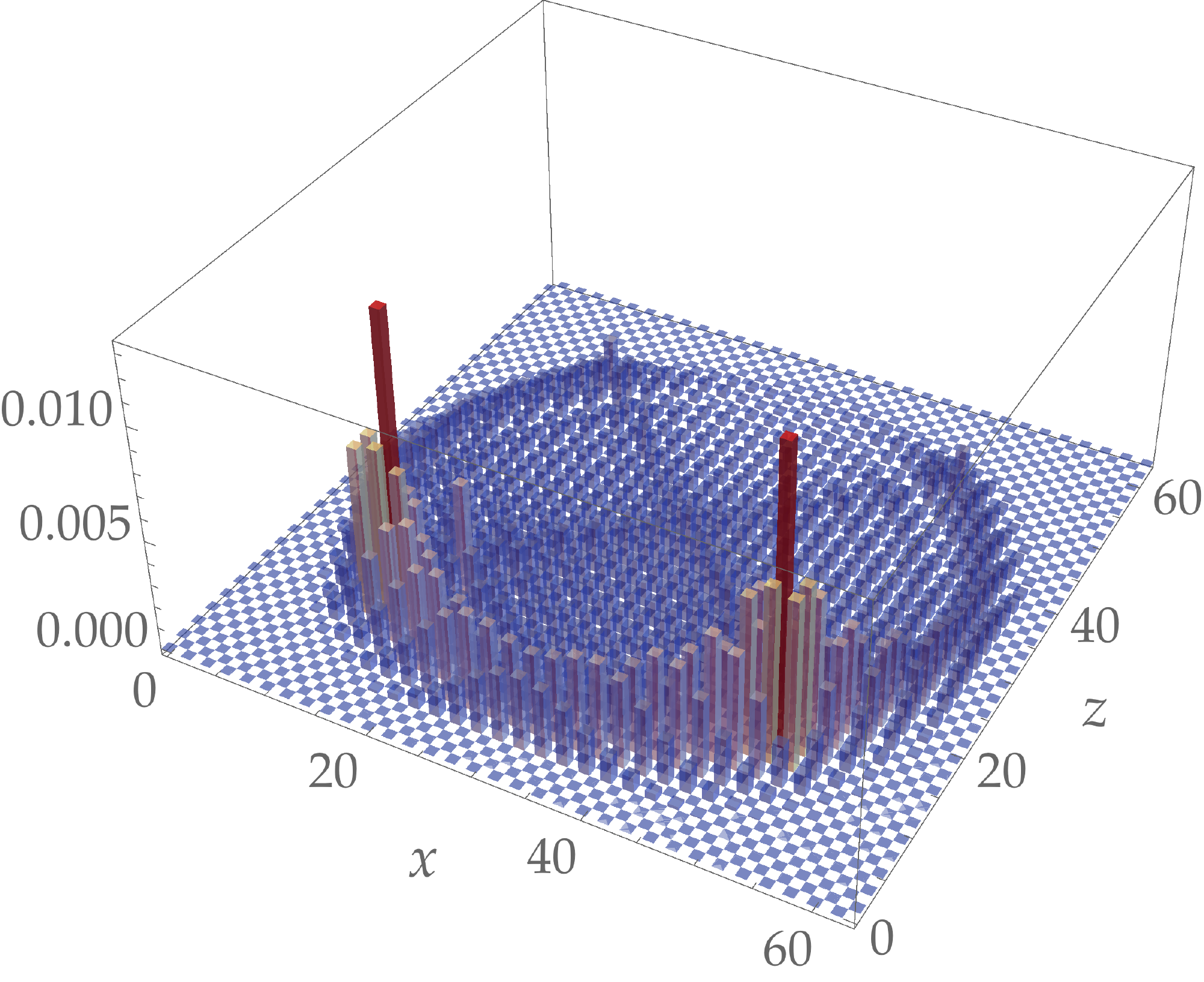}
	\end{subfigure}
	\quad
	\begin{subfigure}[c]{0.25\textwidth}
		\includegraphics[width=\textwidth]{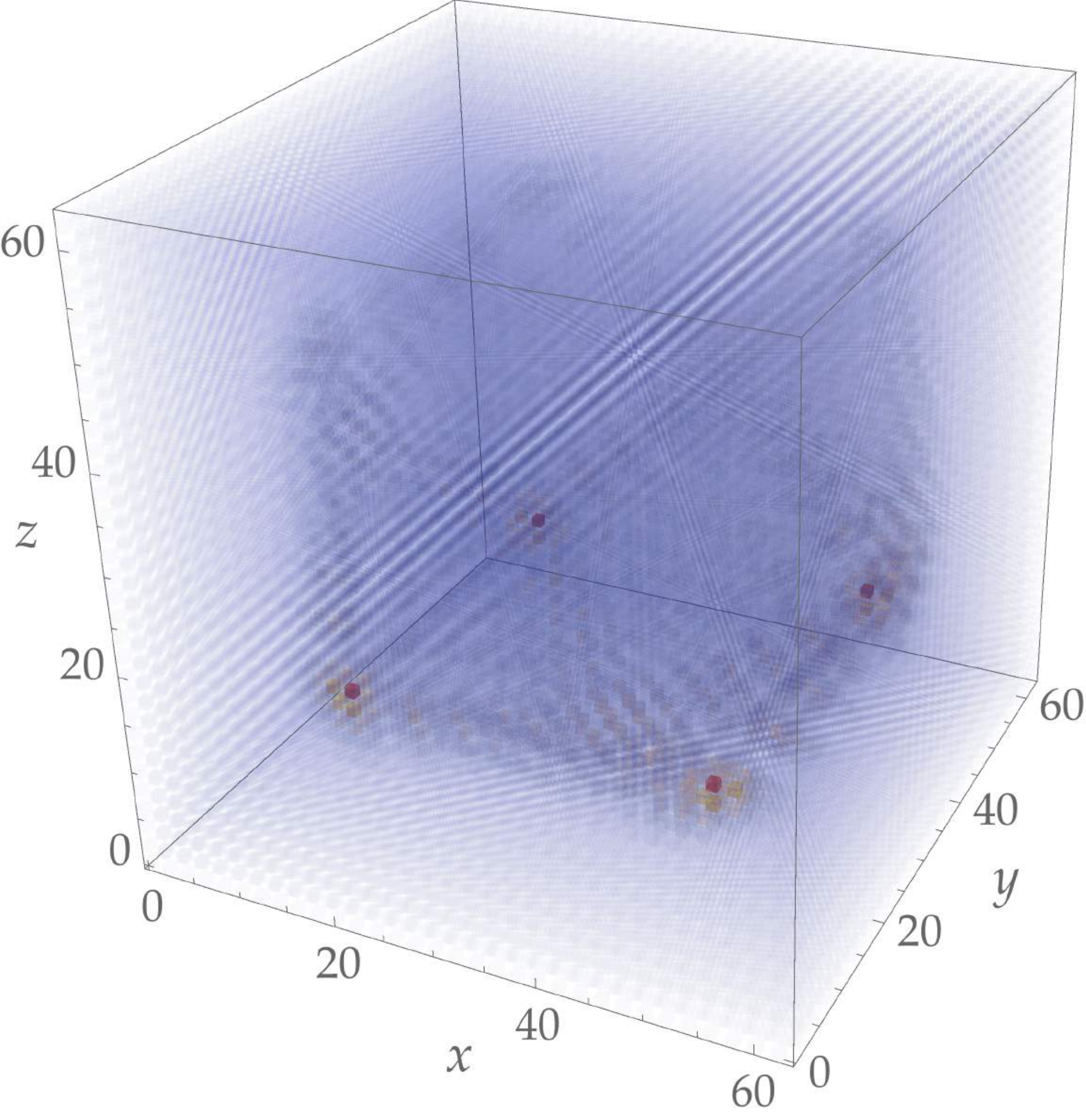}
	\end{subfigure}
	\quad
	\begin{subfigure}[c]{0.25\textwidth}
		\includegraphics[width=\textwidth]{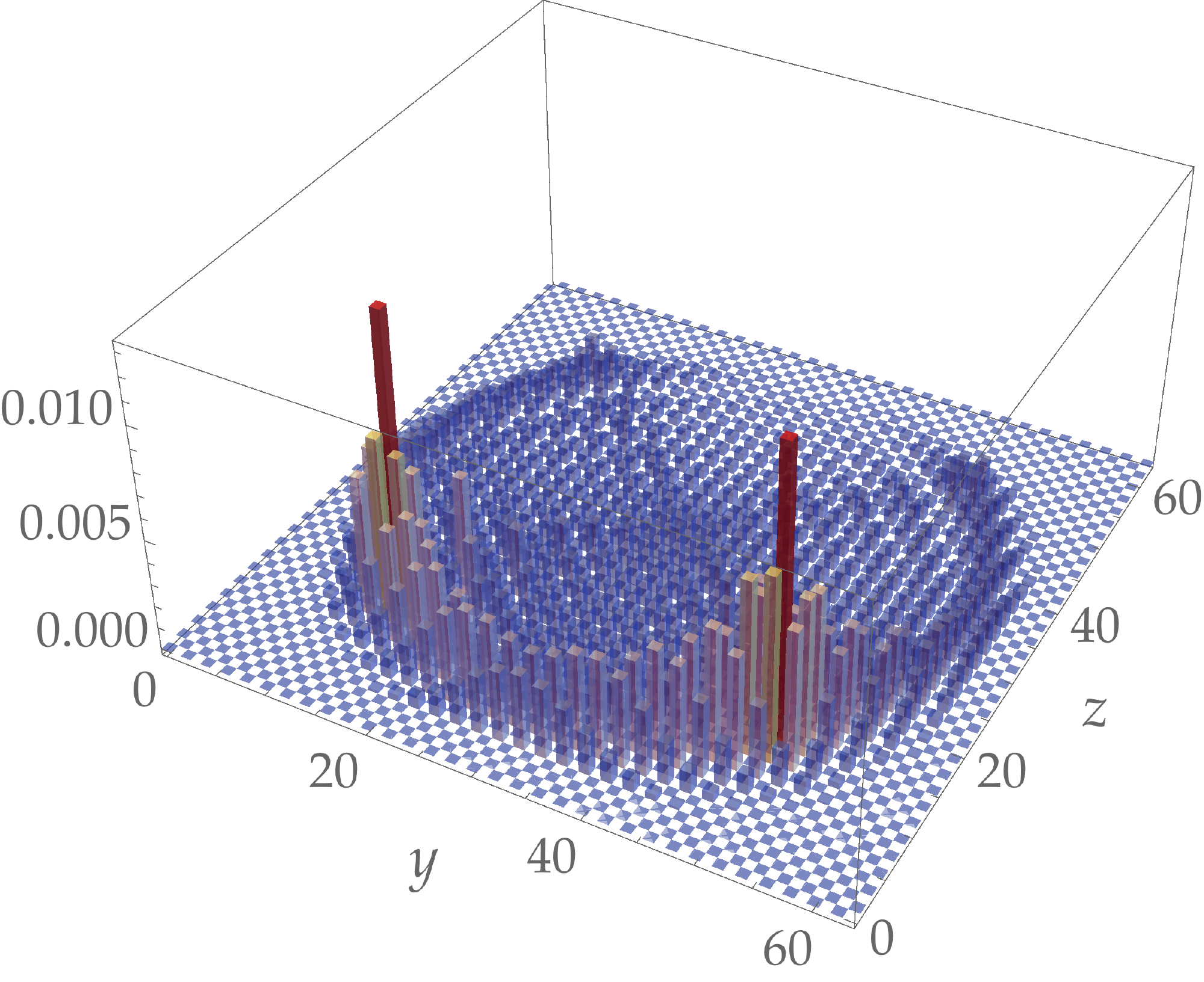}
	\end{subfigure}
	
	\begin{subfigure}[c]{0.25\textwidth}
		\includegraphics[width=\textwidth]{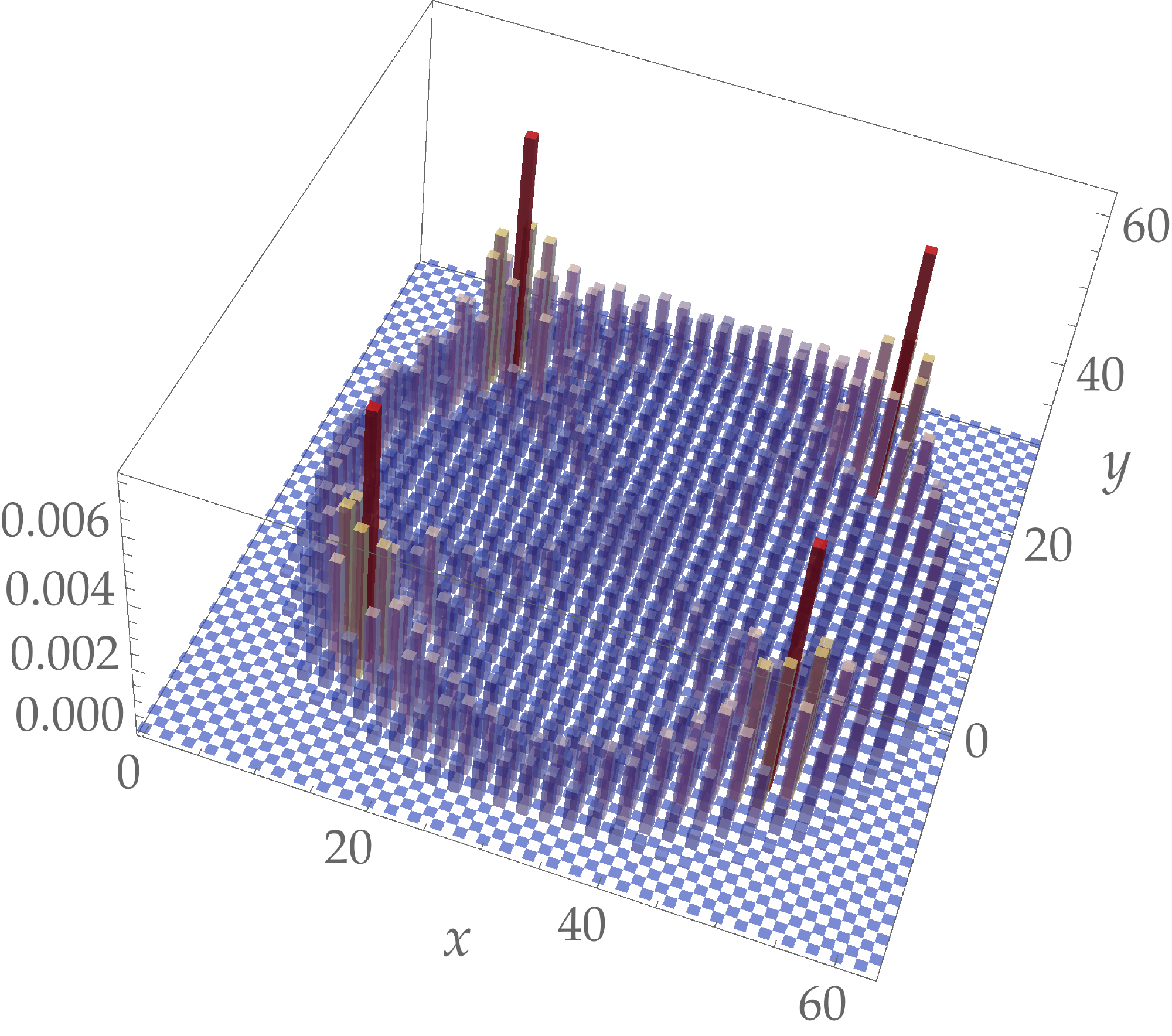}
	\end{subfigure}
	
	\caption{Evolution of a perfectly localised state for the Dirac
		\QW in $3+1$-dimensions. Top-centre: final probability distribution of the same initial state of Figure \ref{fig:loc-3d-3}, after $t=28$ time-steps.
		Top-left: projection of the state along the $y$-axis; top-right: projection along the $x$-axis; bottom: projection on the $(x,y)$-plane.}
	\label{fig:loc-3d-proj}
\end{figure}

The position operator $X$ providing the representation $\ket{\v{x}}$---namely, the
operator such that $X\ket{\v{x}}\ket{\zeta}=x\ket{\v{x}}\ket{\zeta}$---
is $X=\sum_{\v{x}\in\Z^d}\v{x}(\ketbra{\v{x}}{\v{x}}\otimes I)$.
Accordingly, the average position for an arbitrary one-particle state
$\ket{\psi}=\sum_{\v{x},r} g_r (\v{x})\ket{\v{x}}\ket{r}$ is
given by $\braket{\psi|X|\psi}$.

The definition of the mechanical momentum would need an interacting
theory allowing momentum exchange between different
particles. However, in Section \ref{s:dirac-qw}, we have seen that for small $\bk$ and
$m$, the wave-vector $\bk$ (namely the conjugated variable of $\bx$ via
the Fourier transform) corresponds to the Dirac particle
momentum. Moreover, the momentum operator should correspond to the
generator of translations over the lattice. Therefore, as conjugated
momentum we take the following operator
$\v P=\frac{1}{(2\pi)^d}\int_{\mathsf{B}}\dif{\v{k}}\,\v{k}
(\ketbra{\v{k}}{\v{k}}\otimes I)$.
We can now compute the commutator between $X_i$ and $P_j$, $i,j = x,y,z$. That is,
\begin{align}\nonumber
[X_i,P_j]&=\delta_{ij}\frac{1}{2\pi}\sum_{x_i}\int_{-\pi}^{\pi} \dif{k_j}\,\sum_{y_i} y_ik_j \ketbra{x_i}{y_i}
       e^{-ik_j (x_i-y_i)}-\delta_{ij}\frac{1}{2\pi}\int_{-\pi}^{\pi}\dif{k_j}
           \,\sum_{z_i,w_i} z_i k_j  \ketbra{w_i}{z_i} e^{ik_j (z_i-w_i)}\\
     &=\delta_{ij}\frac{1}{2\pi}\sum_{x_i,y_i} \int_{-\pi}^{\pi}\;\dif{k_j} \,
       (x_i-y_i) k_j  \ketbra{x_i}{y}e^{-ik_j(x_i-y_i)},
\end{align}
where in the second equality it was possible to interchange the sum
and the integral according to the Fubini Theorem. Integrating by parts
we get
\begin{align}\label{eq:commutator}
  \bra{\psi}[X_i,P_j]\ket{\psi}=i\left(1-\frac{1}{2}\sum_{r}\left(|\hat {
  g}_{r}(\pi)|^2+ |\hat {g}_{r}(-\pi)|^2\right)\right)\delta_{ij},
\end{align}
where $\ket{\psi}=\sum_{\v{x},r} g_\nu (\v{x})\ket{\v{x}}\ket{r}$
is a generic state and $g(\v{k})$ the discrete Fourier transform of
$g(\v{x})$.  We notice that Equation~\eqref{eq:commutator} differs from the
usual canonical commutation relation by a boundary term, in agreement
with the existence of perfectly localized states for the walk
$\ket{\v{x}}\ket{\zeta}=\sum_{\v{y}\,r}
g_r(\v{y})\ket{\v{y}}\ket{r}$,
$g_r(\v{y})=c_{r}\delta_{\v{x}\v{y}}$, for which the expectation
value in Equation~\eqref{eq:commutator} vanishes.  In the following
evolution of the position expectation value, we will consider states
having negligible boundary term in Equation~\eqref{eq:commutator}.

The evolution of the position operator $\v X(t)={U}^{ - t} {\v X}{U}^{t} $
can be computed via the velocity and the acceleration operators
derived by the commutator with the walk Hamiltonian
\begin{align}
\v{V}(t)=i[H,\v X(t)]=\int_{\mathsf{B}} \dif{\v{k}}\,
  \ketbra{\v{k}}{\v{k}}\otimes \v V(\v{k}) ,\quad
\v{A}(t)=i[H,\v V(t)]=\int_{\mathsf{B}} \dif{\v{k}}\,
  \ketbra{\v{k}}{\v{k}}\otimes \v A(\v{k}),
\end{align}

From direct computation (and neglecting the boundary terms of the
commutators), it follows
\begin{align}
V_j(\v{k})&=\frac{\sin\omega_{\v{k}}-\omega_{\v{k}}\cos\omega_{\v{k}}}{\omega_{\v{k}}\sin\omega_{\v{k}}}H(\v{k})(\v{v}_{\v{k}})_j+\frac{\omega_{\v{k}}}{\sin\omega_{\v{k}}}
            n\gamma_0\mat{\gamma}\cdot\partial_{k_j}\tilde{\v{n}}_{\v{k}},
\end{align}
and
\begin{align}
A_j(\v{k})=2n\frac{\omega^2_{\v{k}}}{\sin^2\omega_{\v{k}}}\left(
n\sum_{\mu<\nu}\gamma_\mu\gamma_\nu f^{(j)}_{\mu\nu}-m\mat{\gamma}\cdot\partial_{k_j}\tilde{\v{n}}_{\v{k}}\right)\\
f^{(j)}_{\mu\nu}\coloneqq\left((\tilde{\v{n}}_{\v{k}})_\nu\partial_{k_j}(\tilde{\v{n}}_{\v{k}})_\mu-(\tilde{\v{n}}_{\v{k}})_\mu\partial_{k_j}(\tilde{\v{n}}_{\v{k}})_\nu\right).
\end{align}

Now we can derive the analytical expression of $\v X(t)$ by doubly
integrating the acceleration operator $\v A(t)$, with
$\v A(\v{k},t)=e^{iH(\v{k})t} \v A(\v{k})e^{-iH(\v{k})t}$. A lengthy but
simple computation (notice that
$n_3 f^{(j)}_{12}-n_2 f^{(j)}_{13}+n_1 f^{(j)}_{23}=0$) shows that
$[H(\v{k}),\v A(\v{k})]_+=0$, which gives
\begin{align}\label{eq:acc}
\v A(\v{k},t)=e^{2iH(\v{k})t} \v A(\v{k}).
\end{align}

Therefore, integrating the first time we get
\begin{equation}\label{eq:Vt}
\begin{aligned}
    \v V(\v{k},t) & = \hat{\v V}(\v{k})+ \v Z^V(\v{k},t),\\
     \hat{\v V}(\v{k}) & = \v V(\v{k})- \v Z^V(\v{k},0), \\ \v Z^V(\v{k},t) & = \frac{1}{2\im} H^{-1}(\v{k}) \v A(\v{k},t),
%-(\v{v}_{\v{k}})_j^2\gamma_0\mat\gamma\cdot\v{n}_{\v{k}} +\gamma_0 (\v{v}_{\v{k}})_j\sqrt{1-(\v{v}_{\v{k}})_j^2}, 
\end{aligned}
\end{equation}
with $H^{-1}(\v{k})=\omega_{\v{k}}^{-2}\,H(\v{k})$, and integrating again one has
\begin{equation}\label{eq:Zx}
\begin{aligned}
    \v{X}(t) & = \v{X}(0)+\hat{\v{V}} t+\v{Z}^{{X}}(t)-\v{Z}^{{X}}(0),\\   
    \v Z^X(\v{k},t) & = -\frac{1}{4}H^{-2}(\v{k}) \v A(\v{k},t),
\end{aligned}
\end{equation}
where 
\begin{equation}
\begin{aligned}
  Z^X_j(\v{k},t) & =
      -\frac{\omega_{\v{k}}^3}{2\sin\omega^3_{\v{k}}}e^{2iH(\v{k})t}
      \left(n^2\gamma_0\mat\gamma\cdot\v{w}^{(j)} +
          nm\tilde{\v{n}}\cdot\partial_{k_j}\tilde{\v{n}} +
          m^2\gamma_0\mat\gamma\cdot\partial_{k_j}\tilde{\v{n}}\right),\\
  \v{w}^{(j)} & \coloneqq
      \begin{pmatrix}
          n_3f_{13}^{(j)}+n_2f_{12}^{(j)}\\
          -n_1f_{12}^{(j)}+n_3f_{23}^{(j)}\\
          -n_1f_{13}^{(j)}+n_2f_{23}^{(j)}
      \end{pmatrix}.
\end{aligned}
\end{equation}

The operator $\hat{\v V}$ in Equations \eqref{eq:Vt} and \eqref{eq:Zx} is the
classical component of the velocity operator which, in the Hamiltonian
diagonal basis Equation \eqref{eq:fw}, is proportional to the group velocity
$\hat{\v V}(\v{k})\propto(\sigma_z\otimes I)\v{v}_{\v{k}}$.
In addition to the classical contribution $\hat{\v V}t$, we see that the
position operator Equation \eqref{eq:Zx} presents, as in the usual Dirac
theory, a time-dependent component $\v Z^X(t)$ and a constant shift term
$\v Z^X(0)$. Since have
\begin{align}
    {}_k\!\braket{\mp| \hat{V}_j(\v{k}) |\pm}_k = 0,\qquad
    {}_k\!\braket{\pm| Z^X_j(\v{k}) |\pm}_k = 0,
\end{align}
the position operator $\v X(t)$ Equation \eqref{eq:Zx} mean value for the generic
state $\ket{\psi}=\ket{\psi_+}+\ket{\psi_-}$ having both particle and
positive and negative frequency components can always be written as
\begin{align}\label{eq:meanpos}
    \braket{\psi | \v X(t) | \psi} & =
        \v x_{\psi}^+(t) + \v x_{\psi}^-(t) + 
        \v x_{\psi}^{\rm{int}}(t)\\ \label{eq:classical-meanpos}
    \v x_{\psi}^{\pm}(t) & \coloneqq 
        \braket{\psi_\pm | \v X(0) + \hat{\v V}t | \psi_\pm}  \\ \label{eq:zitterbewegung}
    \v x_{\psi}^{\rm{int}}(t) & \coloneqq 
        2 \Re\sparen*{\braket{\psi_+ | \v X(0) - \v{Z}^{{X}}(0) + \v{Z}^{{X}}(t) | \psi_-}},
\end{align}
with $\Re$ denoting the real part. The first two terms
$\v x_{\psi}^{\pm}(t)$ simply correspond to the ``classical'' evolution
of the particle and antiparticle components of the initial state
$\ket{\psi}$, which evolve independently according to the classical
component $\hat{\v V}$ of the velocity operator. The interference between
positive and negative frequencies is responsible for the term
$\v x_{\psi}^{\rm{int}}(t)$ in Equation \eqref{eq:zitterbewegung}. Obviously,
in case of $\ket{\psi}$ having only positive or negative component,
the interference disappears. The additional term
$\v x_{\psi}^{\rm{int}}(t)$ consists of two {contributions:} 
 a constant shift
and a time dependent term. 

Taking for example a superposition of
particle and antiparticle states (see Definition
\ref{d:particle-state}),
\begin{equation}\label{eq:state-zitter}
\begin{aligned}
&c_+\ket{\psi_+}+c_-\ket{\psi_-},\qquad |c_+|^2 + |c_-|^2=1,\\
&\ket{\psi_\pm}=\frac{1}{(2\pi)^{d/2}} \int_{\mathsf{B}}\dif{\v{k}}\,
g_{{\v{k}'}}(\v{k},0) \ket{\v{k}}\ket{u_{\pm,p}(\v{k})}, \qquad
D_{\v{k}}\ket{u_{\pm,p}(\v{k})} = e^{\mp i\omega_{\v{k}}}\ket{u_{\pm,p}(\v{k})},
\end{aligned}
\end{equation}
where $\ket{u_{\pm,p}(\bk)}$ are the Dirac walk eigenvectors of Equation (\ref{eq:evecs-D}).
One can show that the time dependent contribution is an oscillating
term that for $t \to \infty$ goes to $0$ as $1/\sqrt{t}$, and whose
amplitude is bounded by $1/m$---say by the Compton wavelength
$\hbar/ m c$ in the usual dimensional units (see
Ref~\cite{Bisio:2013ab} for the proof in one space
dimension). Accordingly, $\v x_{\psi}^{\rm{int}}(t)$ can be considered as
the \QW analogue of the so-called \emph{Zitterbewegung}.

In Figures~\ref{fig:dirac-1d-zitter} and \ref{fig:zitter-3d}, we show two
numerical examples (see Section \ref{s:numerical}) of mean position
evolution for the Dirac \QW in one and three space dimensions,
respectively. In the first case, one can also notice the time-damping
of the jittering amplitudes.

\begin{figure}[H]
    \centering
    \begin{subfigure}[c]{0.35\textwidth}
        \includegraphics[width=\textwidth]{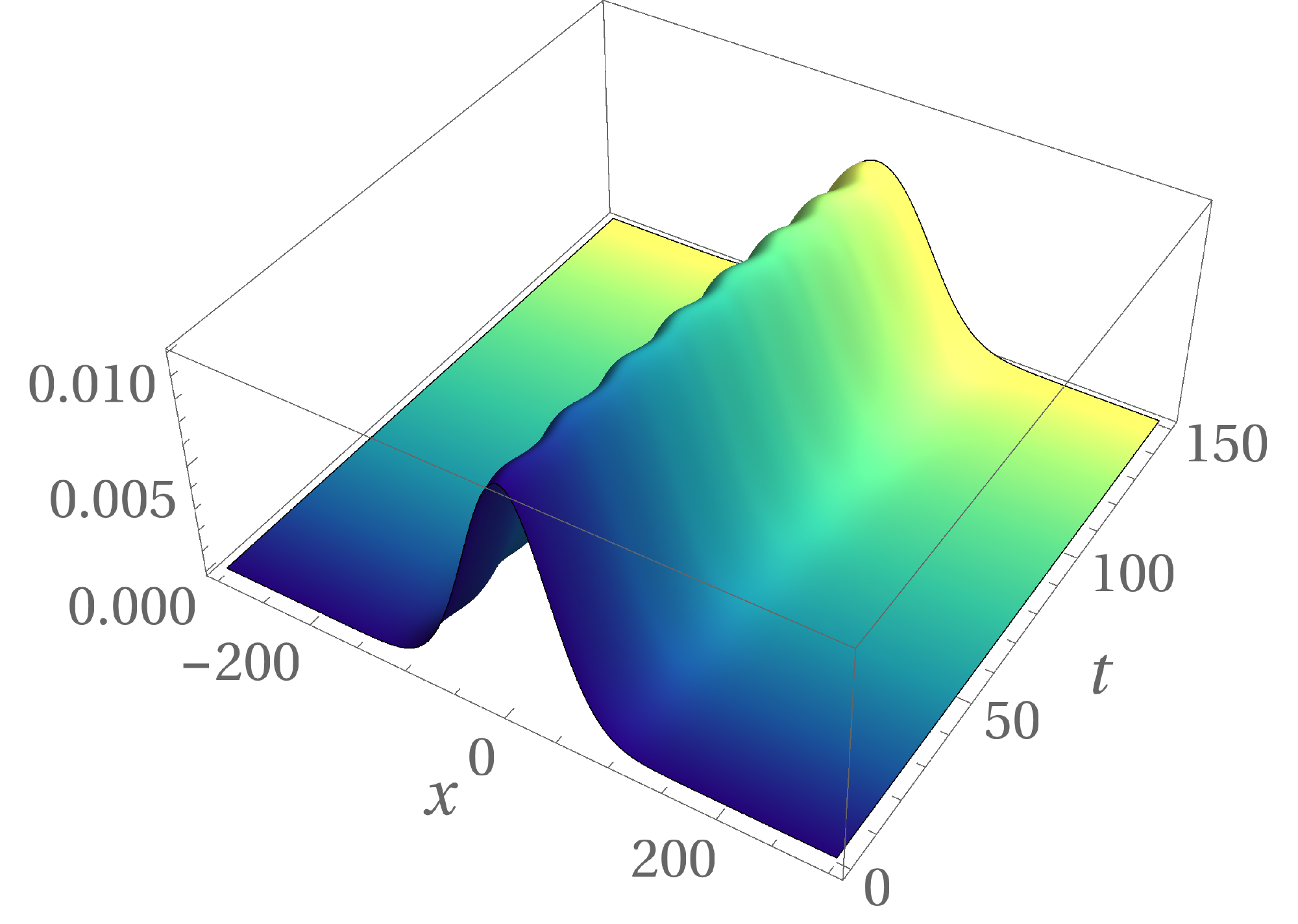}
    \end{subfigure}
    
    \begin{subfigure}[c]{0.35\textwidth}
        \includegraphics[width=\textwidth]{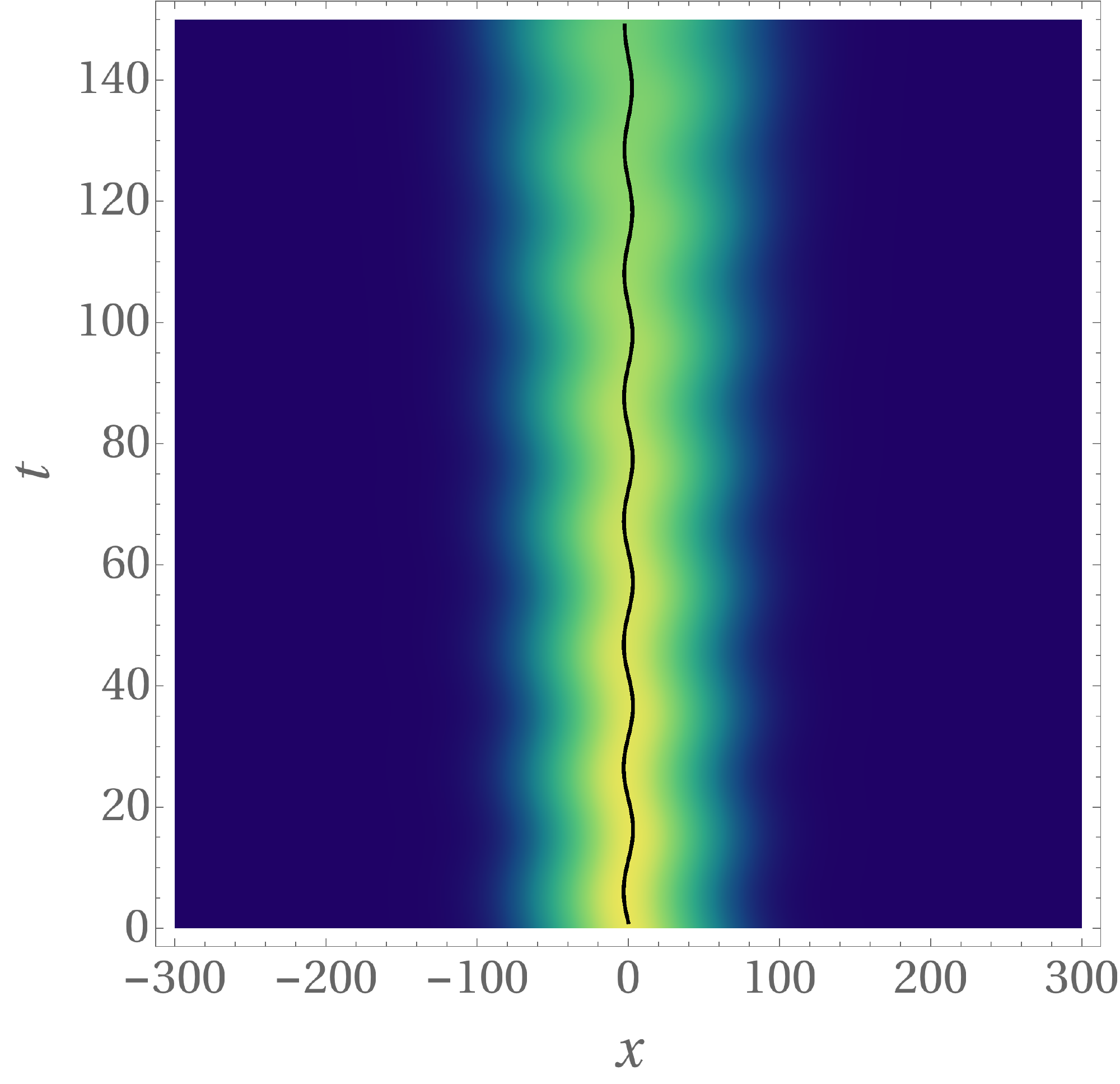}
    \end{subfigure}
    \qquad\quad
    \begin{subfigure}[c]{0.35\textwidth}
        \includegraphics[width=\textwidth]{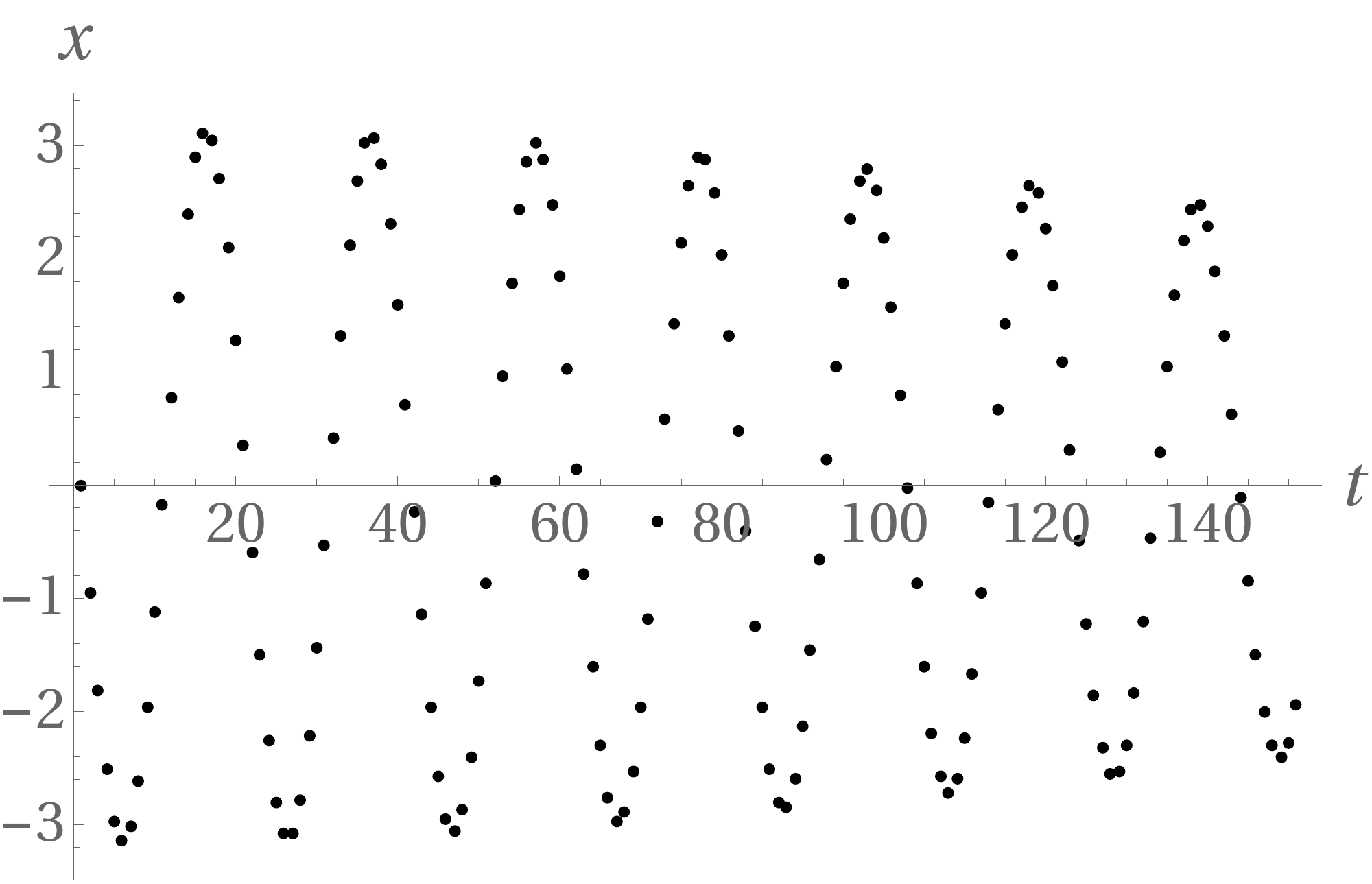}
    \end{subfigure}
    \caption{Evolution according to the Dirac \QW in $1+1$ dimensions for $t = 150$ time-steps of particle states having both a particle and an antiparticle component, as defined in Equation (\ref{eq:state-zitter}). Here the states are Gaussian with parameters: mass $m = 0.15$, width $\sigma = 40^{-1}$, mean wave-vector $k' = 0.01 \pi$, $c_+ = c_- = 1/\sqrt{2}$. Top and bottom-left: probability distribution of the position. Bottom-right: evolution of position mean value.}
    \label{fig:dirac-1d-zitter}
\end{figure}

\begin{figure}[H]
    \centering
    \includegraphics[width=0.4\textwidth]{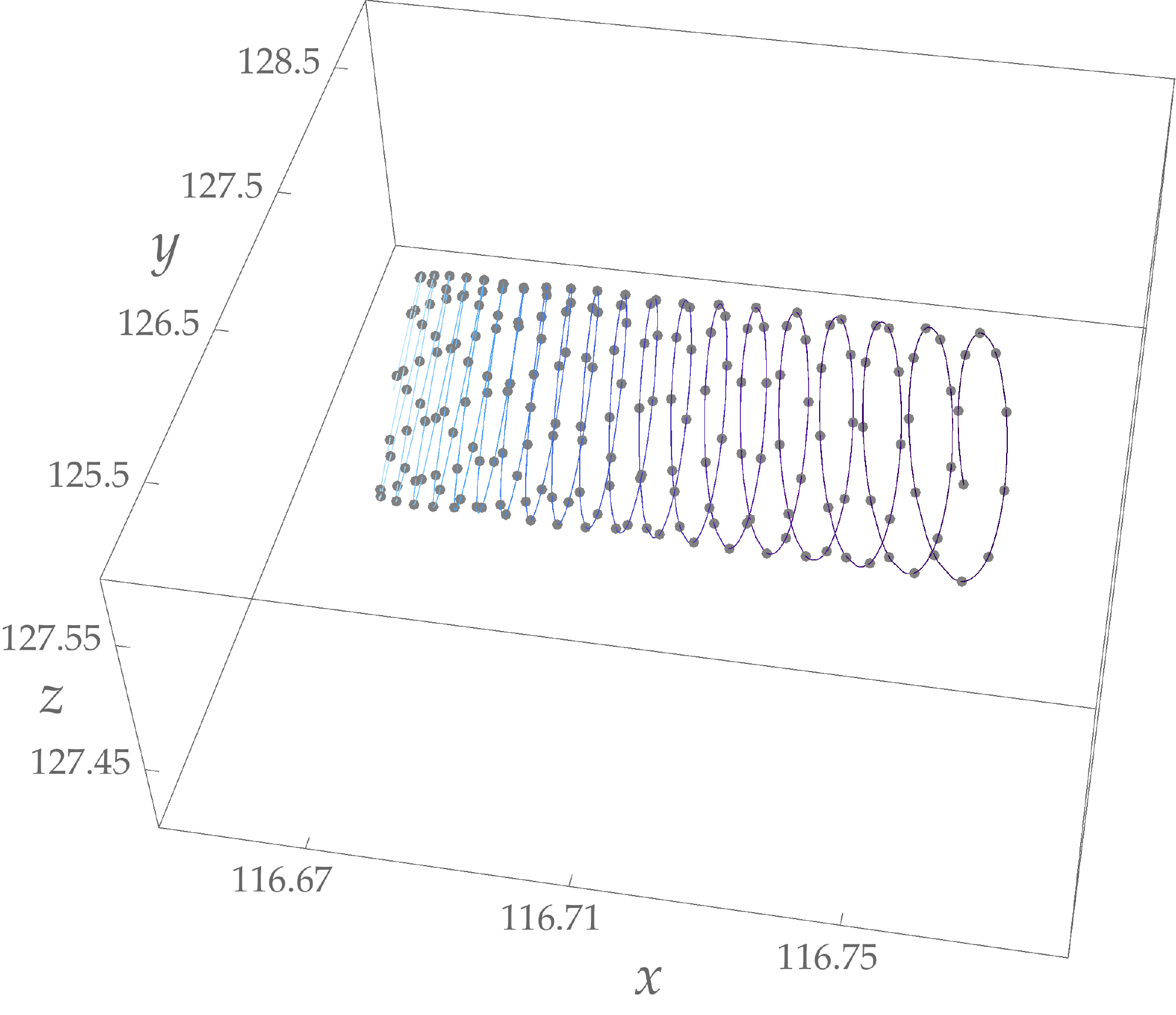}
    \caption{
        Evolution for the mean position according to the Dirac \QW in $3+1$ dimensions for $t = 200$ time-steps of particle states having both a particle and an antiparticle component, as defined in Equation (\ref{eq:state-zitter}). Here the states are Gaussian with parameters:
        mass $m = 0.3$, mean wave-vector $\vec k' = (0,0.01 \pi,0)$, width $\sigma_i = \sigma = 32^{-1}$ for $i = x,y,z$; the spinor components in the walk eigenbasis are $(1/\sqrt{2},0,1/\sqrt{2},0)$, with the first two components corresponding to the positive energy part and the second two to the negative one; time evolution from left to right.}\label{fig:zitter-3d}
\end{figure}

\begin{Remark}[Newton--Wigner position operator evolution]
As in QFT, one can define the {\em Newton--Wigner} position
    operator $\v X_{\textup{NW}}$ which does not mix states with positive and negative
    eigenvalues. Given the operator $W_{\textup{FW}}$ providing the {\em
      Foldy--Wouthuysen} representation of the Dirac walk, namely
    the representation in which the Hamiltonian $H(\v{k})$ is diagonal
\begin{align}\label{eq:fw}
  W_{\textup{FW}}=\int_{\mathsf{B}}\dif{\v{k}}\, \ketbra{\v{k}}{\v{k}}\otimes W_{\textup{FW}}(\v{k}),\qquad   W_{\textup{FW}}(\v{k})\colon\{\ket{\nu}\}\rightarrow \{\ket{u}_{\v{k}}\},\\
  W^{-1}_{\textup{FW}}(\v{k})H(\v{k})W_{\textup{FW}}(\v{k})=\mathrm{diag}(\omega_{\v{k}},\omega_{\v{k}},-\omega_{\v{k}},-\omega_{\v{k}}),
\end{align}
the Newton--Wigner rotated position operator is defined as
\begin{align}\label{eq:NW}
  \v X_{\textup{NW}}=W^{-1}_{\textup{FW}} \v X W_{\textup{FW}}.
\end{align}

As in the usual \QFT, the Newton--Wigner position operator Equation \eqref{eq:NW}
does not suffer the jittering of the mean position even for states
having both a particle and an antiparticle component. Indeed, in this
case, the velocity operator
    \begin{align}\label{eq:NW-Vt}
      \v V_{\textup{NW}}(t)=i[H,\v X_{\textup{NW}}(t)],\qquad \v V(k)=\hat{\v V}(\v{k}),
    \end{align} 
    corresponds to the classical component of the
    velocity operator in Equation~\eqref{eq:Vt} and leads to a null
    acceleration $\v A(t)=i[H,\v V_{\textup{NW}}(t)]=0$. By integrating
    Equation \eqref{eq:NW-Vt}, we see that the time evolution of the
    Newton--Wigner position operator $\v X_{\textup{NW}}(t)$ is simply
\begin{align}
  \v X_{\textup{NW}}(t)=\v X_{\textup{NW}}(0)+\hat{\v V}t.
\end{align}
\end{Remark}

\section{Conclusions}

The QW framework, say a lattice of quantum systems in local unitary
interaction, appears to be very promising both from the
information-theoretical perspective, in that QWs can be exploited to
solve efficiently some search problems, and for the connection
existing between the a discrete time quantum walk evolution and the
relativistic equations of motion. In this paper, we analyse both
numerically and analytically the properties of QWs on Abelian lattices
up to $3+1$ dimensions.  The Weyl QWs considered  here are the only
isotropic (all the directions on the lattice are equivalent) QWs
admissible on Abelian lattices and with two-dimensional coin
system. The QWs in one and two space dimensions are defined on the
simple cubic lattice while the QW in $3+1$ dimensions is defined on
the body-centered cubic lattice. As shown in
Ref.~\cite{DAriano:2014ae}, any other topology fails to accommodate a
non-trivial QW (by trivial QW we mean a walk corresponding to the
identical evolution or to a shift in a fixed direction). The only
coupling of two Weyl QWs that preserves locality is then defined Dirac
QW.  Remarkably, the selected walks are compatible with a ``large
scale'' relativistic dynamics.

The analytical results of this paper show that for particle states as
defined in \cref{d:particle-state}, the Weyl and Dirac QW dynamics is
well approximated by a dispersive differential equation whose drift
and diffusion coefficients reduce to the usual Weyl and Dirac ones in
the limit of small wave-vectors.  The numerical results are the first
simulations of QWs in $3+1$ dimensions and on the BCC lattice.  The
numerical results are given for the Dirac QW in $1+1$ and $3+1$
dimensions for different types of initial states. In $3+1$ dimensions
we show the evolution of both particle states and perfectly localised
states.  In $1+1$ dimensions, the evolution of the superposition of
positive and negative energy states for the Dirac QW produces (as
depicted in Figure \ref{fig:dirac-1d-zitter}) the well-known
\emph{Zitterbewegung} effect of the relativistic electron.  The
appearance of this oscillating phenomenon is also shown for the Dirac
QW in $3+1$ dimensions (see Figure \ref{fig:zitter-3d}).

As already mentioned, the QW framework can accommodate {from} 
a
theoretical viewpoint a local discrete time unitary evolution as the
{microscopic} 
 description of relativistic {particle} 
  dynamics. The last
one is obtained as an approximation of the QW evolution for a specific
class of quantum states; namely, states narrow-banded in small
wave-vectors (see Definition \ref{d:particle-state}). The same QW on
arbitrary states (for example, localized states) shows a very different
dynamical behaviour that cannot be interpreted as a particle
evolution. While the approximation of Proposition
\ref{p:semi-classical-states} only works for narrow-banded states, the
numerical analysis presented in the manuscript applies to arbitrary
states.

Our results agree with other works (see for example
Ref.~\cite{Strauch:2006aa}) in one space dimensions that studied the
continuum limit of QWs, namely the lattice spacings and the time steps
are sent to 0, in comparison with the Dirac or the Klein–Gordon
equations. Here we do not take the same continuum limit but show that
for specific input states the QW evolution recovers the relativistic
one. Moreover, we do not consider only one space dimension, but also
the two and three space dimensional case where the notion of spin
becomes relevant.

Discrete time QWs provide a local and unitary evolution underlying the
relativistic dynamics and do not start from a finite difference
counterpart of the relativistic differential equations (or
Hamiltonians). The main difference is in the notion of locality, since
the locality of the Hamiltonian does not correspond to the locality of
the unitary operator and {\emph{vice versa}}. As a consequence the “effective”
Hamiltonian corresponding to the Weyl (Dirac) QW differs from the
usual Weyl (Dirac) finite difference Hamiltonian (see the
$\mathrm{sinc}$ function that appears in
Equations \eqref{eq:weyl-interpolating} and
\eqref{eq:dirac-interpolating}). In the limit of small wave-vectors,
the two Hamiltonians coincide, and both give the usual relativistic
dynamics. However, for large wave-vectors they differ significantly.

The Weyl and Dirac QWs presented in this paper also provide an
alternative way to discretize the usual Weyl and Dirac dynamics. The
numerical results of this manuscript can be compared with other
numerical approaches in the literature; see for example
Ref.~\cite{Bauke20112454,thaller2004visualizing,PhysRevA.59.604,Mocken2008868,FillionGourdeau20121403},
where the authors adopt split-operator schemes to approximate the
solutions of the Weyl and Dirac differential equations and recover the
usual relativistic dynamics in the continuum limit.

%%%%%%%%%%%%%%%%%%%%%%%%%%%%%%%%%%%%%%%%%%
\vspace{6pt}  %%MDPI internal note: new layout%%
%% optional
%\supplementary{\textbf{Supplementary Materials:} The following are available online at www.mdpi.com/link, Figure S1: title, Table S1: title, Video S1: title.}  %%MDPI internal note: new layout%%

%%%%%%%%%%%%%%%%%%%%%%%%%%%%%%%%%%%%%%%%%%

\acknowledgments{\textbf{Acknowledgments:} This work has been supported in part by the Templeton Foundation
    under the project ID\# 43796 A Quantum-Digital Universe.}
    \vspace{6pt}

%%%%%%%%%%%%%%%%%%%%%%%%%%%%%%%%%%%%%%%%%%

\authorcontributions{\textbf{Author Contributions:}
Giacomo Mauro D'Ariano conceived the project. Nicola Mosco implemented all the numerical simulations writing original codes in C++ and in Wolfram Mathematica. Paolo Perinotti and Alessandro Tosini provided the analytical analysis. Correspondence and requests for materials should be addressed to all the authors. All authors have read and approved the final manuscript.}

\conflictofinterests{\textbf{Conflicts of Interests:} {The authors declare no conflict of~interest}.

} 

\bibliographystyle{mdpi}
\renewcommand\bibname{References}

%%%%%%%%%%%%%%%%%%%%%%%%%%%%%%%%%%%%%%%%%%
%% optional
%%MDPI internal note: new layout%%
%\sampleavailability{\textbf{Sample Availability:} Samples of the compounds ...... are available from the authors.}
%%%%%%%%%%%%%%%%%%%%%%%%%%%%%%%%%%%%%%%%%%

\end{document}